\documentclass[11pt, oneside]{article}   	
\usepackage{geometry}                		
\usepackage{graphicx}	
\usepackage{mathtools}
\usepackage{amsmath,amssymb,amsthm}
\usepackage{mathrsfs}
\usepackage[colorlinks=true,
            linkcolor=blue,
            urlcolor=blue,
            citecolor=blue]{hyperref}
\usepackage{tikz}
\usepackage{comment}

\title{Characterizing normality via automata and random matrix products}
\author{Laurent Bienvenu, Santiago Cifuentes, Hugo Gimbert}

\newtheorem{theorem}{Theorem}
\newtheorem{definition}[theorem]{Definition}
\newtheorem{proposition}[theorem]{Proposition}
\newtheorem{lemma}[theorem]{Lemma}
\newtheorem{claim}[theorem]{Claim}
\newtheorem{corollary}[theorem]{Corollary}
\newtheorem{example}{Example}

\newcommand{\emptystr}{\Lambda}
\renewcommand{\alph}{\mathrm{A}}
\newcommand{\norm}[1]{\ensuremath{\|#1\|}}
\renewcommand{\u}{\mathbf{u}}
\renewcommand{\v}{\mathbf{v}}
\newcommand{\s}{\mathbf{s}}
\newcommand{\x}{\mathbf{x}}
\newcommand{\y}{\mathbf{y}}
\newcommand{\w}{\mathbf{w}}
\newcommand{\e}{\mathbf{e}}
\newcommand{\zero}{\mathbf{0}}
\newcommand{\one}{\mathbf{1}}
\newcommand{\cM}{\mathcal{M}}
\newcommand{\N}{\mathbb{N}}
\newcommand{\M}{\mathbf{M}}
\newcommand{\mat}[3]{\cM_{#1 {\scriptscriptstyle \times} #2}\left(#3\right)}
\newcommand{\genvect}{\mat{1}{m}{\mathbb{R}}}
\newcommand{\vect}{\mathbb{R}_{\geq 0}^m}
\newcommand{\pvect}{\mathbb{R}_{> 0}^m}
\newcommand{\matm}{\mat{m}{m}{\mathbb{R}_{\geq 0}}}

\renewcommand{\dh}{\ensuremath{d_{H}}}
\newcommand{\U}{\mathbf{U}}
\renewcommand{\P}{\mathbb{P}}
\newcommand{\supp}{\mathrm{supp}}
\newcommand{\interval}{{\ensuremath{[1,\ldots,m]}}}
\newcommand{\ps}{2^\interval}

\newcommand{\Live}{\mathrm{Live}}
\newcommand{\A}{\mathfrak{A}}
\newcommand{\B}{\mathfrak{B}}

\newcommand{\D}{\mathbf{D}}
\renewcommand{\S}{\mathcal{S}}
\newcommand{\R}{\mathbb{R}}
\newcommand{\T}{\mathbf{T}}
\newcommand{\PSPACE}{\textsc{PSPACE}}
\newcommand{\betfunc}{\gamma}

 \usepackage{xcolor}

\usepackage{caption}
\begin{document}
\maketitle

\begin{abstract}
For a fixed alphabet $\alph$, an infinite sequence $X$ is said to be normal if every word $w$ over $\alph$ appears in $X$ with the same frequency as any other word of the same length. A classical result relates normality to finite automata as follows: a sequence $X$ is normal if and only if all gambling strategies implementable with finite deterministic automata lose all their risked capital when trying to predict the next bit of $X$ after seeing the ones before. More precisely, Schnorr and Stimm (1972) proved that the capital goes exponentially fast to zero unless the gambler represented by the automaton eventually stops betting, in which case the capital eventually remains constant. In this paper we show that an analogous result holds when considering probabilistic automata: a sequence $X$ is normal if and only if for any gambling strategy implementable with a probabilistic finite automaton it holds that the expected value of the capital of the gambler converges exponentially fast to a finite value when playing against $X$. To obtain this result, we show a more general statement related to the convergence of supermartingales given by certain finite families of non-negative matrices $\{\M_a\}_{a\in \alph}$ which we call ``superfair''. In particular, we show that $X$ is normal if and only if for every superfair family of matrices, $||\v \M_{X[1]} \ldots \M_{X[n]}||$ converges exponentially fast to a finite value for any non-negative starting vector $\v$. Moreover, we distinguish three distinct behaviours that this sequence can attain, and prove that the problem of recognizing, given a family of matrices with algebraic entries, to which case it belongs, is decidable. 
\end{abstract}

\section{Introduction}

\subsection{Normality as unpredictability}

Given a finite alphabet $\alph$, an infinite sequence $X$ of elements from $\alph$ is said to be \textit{normal}
if every word appears as a substring of $X$ with the same frequency as any other word of the same length. This
notion was introduced by Borel in 1909 \cite{emile1909probabilites} as a way to formalize statistical randomness at the level of finite substrings.
Since then, normal sequences have been studied in the context of number theory \cite{champernowne1933construction,copeland1946note} and theoretical computer science \cite{agafonov1968normal, becher2002example, schnorr1972endliche}.
It is well-known that almost all sequences are normal (in the sense of Lebesgue measure), and there are explicit constructions, such as the Champernowne constant
\begin{align*}
0.12345678910111213\ldots. 
\end{align*}

The definition of normality is syntactical: it only depends on the frequency of finite substrings. However, it can also be defined, in the spirit of the theory known as ``algorithmic randomness'' (whose goal is to give a meaning to the notion of individual random objects, see~\cite{DowneyH2010,LiV2008,Nies2009,ShenUV2017} for the main textbooks on the subject) via the so-called unpredictability paradigm. Imagine we have a betting strategy -- formalized via the notion of martingale --  that bets on the values of an infinite sequence~$X$ where for each value the bet is based on values already seen (and the reward or loss is fair when the bet is correct or incorrect respectively). Whenever a martingale $D$ is able to obtain over time an unbounded capital against this infinite sequence~$X$, it is said that $D$~\textit{succeeds} against~$X$. Then, a sequence $X$ is said to be ``unpredictable'' relative to some class $\mathcal{C}$ of martingales if no martingale~$D$ in the class~$\mathcal{C}$ succeeds against~$X$. Classes~$\mathcal{C}$ of interest are usually those corresponding to a particular computability class: one can for example take for~$\mathcal{C}$ the class of computable martingales, or primitive recursive martingales, or polynomial-time martingales, etc. As one might expect, different classes~$\mathcal{C}$ yield different notions of randomness and as it turns out, when~$\mathcal{C}$ is the class of martingales representable by finite automata, unpredictable sequences relative to~$\mathcal{C}$ are exactly the normal sequences! This fundamental theorem is due to Schnorr and Stimm~\cite{schnorr1972endliche} and informally tells us that the only type of non-randomness that can be detected by finite automata is non-normality. \\
\begin{theorem}[Schnorr-Stimm~\cite{schnorr1972endliche}]\label{thm:schnorr-stimm}
An infinite sequence~$X$ is normal if and only if no martingale computable by a finite automaton can succeed on it. 
\end{theorem}

In recent years, several papers explored whether giving the martingales access to randomness increases their predictive power. The first paper to provide a framework for this type of question is due to Buss and Minnes~\cite{buss2013probabilistic}, where they considered the set of computable martingales with internal access to random coin flips. When playing with a probabilistic martingale~$D$ (which now becomes a random variable) against a given sequence~$X$, there are several types of success:
\begin{itemize}
\item ($\mathbb{P}$-success) With positive probability (over the internal random source), $D$ succeeds against~$X$
\item ($\mathbb{P}1$-success) With probability~$1$ (over the internal random source), $D$ succeeds against~$X$
\item ($\mathbb{E}$-success) The expectation $\mathbb{E}(D(X[1\ldots n]))$ of the capital of~$D$ after playing against the first~$n$ bits of~$X$ takes arbitrarily large values as $n \rightarrow \infty$. 
\end{itemize}
It is clear from the definition that $\mathbb{P}1$-success implies $\mathbb{P}$-success, and in general no other implication holds.

Buss and Minnes studied $\mathbb{P}1$-success and $\mathbb{E}$-success in the context of computable martingales. They proved the following.

\begin{theorem}[{\cite[Theorems 5.1 and 5.2]{buss2013probabilistic}}]
Let~$X$ be an infinite sequence such that no computable martingale succeeds on it. Then, no probabilistically computable martingale can $\mathbb{P}1$-succeed or $\mathbb{E}$-succeed on it either. 
\end{theorem}

Surprisingly, Bienvenu, Delle Rose and Steifer were able to show the opposite result for $\mathbb{P}$-success. 

\begin{theorem}[\cite{BienvenuDS2022}]
There exists a sequence~$X$ such that no computable martingale succeeds on it, but some probabilistically computable martingale $\mathbb{P}$-succeeds on it. 
\end{theorem}

This result is rather counter-intuitive: it says that there is a sequence~$X$ which cannot be algorithmically predicted, but if somehow we use a random source (which is independent of~$X$, hence should carry no information about it), then suddenly we are able to guess~$X$ with some positive (and in fact, as close to~$1$ as we want) probability. \\

In light of these results, it is natural to ask what the situation is for probabilistic automata: Which sequences can be defeated by martingales computable by a probabilistic automaton, in the sense of $\mathbb{P}$-success, $\mathbb{P}1$-success or $\mathbb{E}$-success? Does the Schnorr-Stimm theorem (Theorem~\ref{thm:schnorr-stimm}) still hold for these stronger prediction models? A slightly different (but roughly equivalent\footnote{They considered a different prediction model, known as selectors, but it is an easy exercise to turn a winning selector into a winning martingale and vice versa.}) version of this question was first considered by Lechine, Seiller and Simonsen in~\cite{LechineSS2024}, where they (essentially) showed that probabilistic automata whose transition probabilities are rational cannot $\P$-succeed against a normal sequence. This was later extended by Bienvenu, Gimbert and Pulari~\cite{BienvenuGP2025} to all probabilistic automata.

\begin{theorem}[\cite{BienvenuGP2025,LechineSS2024}]
A sequence $X$ is normal if and only if no martingale computable by a probabilistic automaton can $\P$-succeed on~$X$ (if and only if no martingale computable by a probabilistic automaton can $\P1$-succeed on~$X$). 
\end{theorem}

In this paper, we ask whether this last theorem also holds for $\mathbb{E}$-success. This is a much more intricate question as the expectation of a martingale computable by a probabilistic automaton cannot in general be itself computed/simulated by an automaton. We reduce the question to a random matrix multiplication problem. Let $\v$ be a non-negative initial vector in $\mathbb{R}^m$, let $\{\M_a\}_{a \in \alph}$ be a family of non-negative $m \times m$ matrices satisfying for all non-negative vectors~$\u$:
\begin{align*}
\norm{\u}\geq \frac{\sum_{a \in \alph} \norm{\u \M_a}}{|\alph|}, 
\end{align*}
(where $\norm{\u} = \sum_i |\u(i)|$) and let~$X=a_1 a_2 \ldots$ be a normal infinite sequence over the alphabet~$\alph$. Then our initial question about $\mathbb{E}$-success amounts to asking whether the norm
\[
\norm{\v \M_{a_1} \M_{a_2} \ldots \M_{a_n}}
\]
remains bounded or not. We prove that it indeed remains bounded, and furthermore converges to a finite limit with exponential speed as $n \rightarrow \infty$. As a corollary, we will obtain that normality also prevents $\mathbb{E}$-success.

The matrix convergence theorem (Theorem~\ref{thm:main}) implies convergence of expected capital (Theorem~\ref{teo:prob_automata_fail_against_normal}); together with the Schnorr-Stimm theorem, this gives the characterization of normality in Corollary~\ref{coro:expected_success_and_normality}.

\subsection{Notation and terminology}

Throughout the paper, $\alph$ is a fixed finite alphabet and $m$ a fixed positive integer (the size of the matrices in the family we'll consider).

We denote by $\alph^*$ the set of words over $\alph$. The empty word is denoted by $\emptystr$. Given $w \in \alph^*$, we write $|w|$ for the length of $w$. For $1 \leq i \leq |w|$, $w[i]$ is the $i$-th letter of $w$, and for $i \leq j \leq |w|$, $w[i \ldots j]$ is the sub-word $w[i] w[i+1] \ldots w[j]$. We write $w \sqsubseteq w'$ to mean that $w$ is a prefix of $w'$, and $w \sqsubset w'$ when $w$ is a strict prefix of $w'$ (that is, $w \sqsubseteq w'$ and $|w|<|w'|$). The set of infinite sequences is denoted by $\alph^\omega$, and we define $X[i]$ and $X[i \ldots j]$ in the same way as for words. We sometimes need to choose a random word or sequence. When we say that we choose a sequence $X \in \alph^\omega$ \emph{uniformly at random}, we mean that each letter $X[i]$ is drawn at random in~$\alph$ with all letters having probability $1/|\alph|$, and for $i<j$, the values of $X[i]$ and $X[j]$ are independent. The probability measure induced by this process on the set $\alph^\omega$ is referred to as the uniform measure, or Lebesgue measure.

The set of $n \times k$ real-valued matrices is denoted by $\mat{n}{k}{\R}$ (and, likewise, $\mat{n}{k}{\R_{\geq 0}}$ or $\mat{n}{k}{\R_{> 0}}$ for non-negative and positive $n \times k$ matrices). For $i,j$ indices, $\M(i,j)$ is the value of the matrix~$\M$ at position $(i,j)$. Most matrices considered below are square $m \times m$ matrices, but for $C,D \subseteq \interval$ of size $n$ and $k$ respectively, $\M^{[C \times D]}$ is the $n \times k$ submatrix of~$\M$, obtained by keeping only the entries $\M(i,j)$ of $\M$ such that $i \in C$ and $j \in D$.

Vectors in $\R^m$ are identified with $1$-line matrices (i.e., elements of $\genvect$) and we write $\v(i)$ for $\v(1,i)$. For $1 \leq i \leq m$, $\e_i$ is the canonical basis vector such that $\e_i(i)=1$ and $\e_i(j)=0$ for $j \not=i$. The zero vector is written~$\zero$. The norm of a vector $\v$ is the quantity $\norm{\v} = \sum_i |\v(i)|$. Observe that when a pair $\v,\v'$ of vectors is nonnegative (as are most vectors in this paper), we have $\norm{\v+\v'}=\norm{\v}+\norm{\v'}$. We define $\U_m$ to be the subset of $\vect$ consisting of vectors of norm~$1$ (a.k.a. \emph{unitary} vectors), and $\U_m^+$ as the subset of $\U_m$ consisting of vectors with positive coordinates.

Given $\v \in \vect$ a nonnegative vector, we call the set $\{i \in \interval \mid \v(i)>0\}$ the \emph{support} of $\v$ and denote it by $\supp(\v)$. Given $E \subseteq \interval$, we denote by $\R^m_{\geq 0}(E)$ the set of nonnegative vectors whose support is exactly~$E$ and by $\U_m^+(E)$ the set of unitary vectors whose support is exactly~$E$. 

As stated before, we will be interested in studying, given some normal sequence $X = a_1 a_2\ldots$ over alphabet $\alph$, the behaviour of the sequence $(\v \M_{a_1} \M_{a_2} \ldots \M_{a_n})_{n \in \mathbb{N}}$ for any initial non-negative vector $\v$ under the assumption that the family of matrices $\{\M_a\}_{a \in \alph}$ is \textit{fair} in the following sense.

\begin{definition}
A family $(\M_a)_{a \in \alph}$ of matrices in $\matm$ is said to be \emph{fair} if for every vector $\v \in \vect$
\[
\norm{\v} = \frac{\sum_{a \in \alph} \norm{\v \M_a}}{| \alph |}. 
\]

It is said to be \emph{superfair} if for every vector $\v \in \vect$
\[
\norm{\v} \geq  \frac{\sum_{a \in \alph} \norm{\v \M_a}}{| \alph |}. 
\]
\end{definition}

Note that, rather awkwardly, fairness implies superfairness, but not the other way around. Nonetheless, we use this terminology to resemble the concepts of martingale and supermartingale that will be defined later.\\

From here on, the superfair family $\{\M_a\}_{a \in \alph}$ is fixed once and for all. For~$w=a_1 \ldots a_k$, we write $\M_w$ for the product $\M_{a_1} \ldots \M_{a_k}$ (and $\M_\emptystr$ is the identity matrix). 

Our main theorem is the following: 

\begin{theorem}\label{thm:main}
Let $(\M_a)_{a \in \alph}$ be a superfair family of matrices. If $X \in \alph^\omega$ is a normal sequence and $\v \in \vect$ a nonnegative vector,  then the sequence $(\norm{\v \M_{X[1 \ldots n]}})_{n \in \N}$ converges to a finite limit with exponential speed. 
\end{theorem}
\noindent (where we say that a sequence of reals $(x_n)$ converges exponentially fast towards a limit~$l$ if $|x_n-l| = O(e^{-\alpha n})$ for some $\alpha >0$). The rest of the paper will be devoted to the proof of this theorem. We will first show it with the following extra assumption:

\begin{center}
$(\bigstar)$ For any $i,j \in \interval$, there exists $w \in \alph^*$ such that $\M_w(i,j)>0$. 
\end{center}

Then, we will see how to remove the assumption to treat the general case. But before proving this theorem, let us see how to apply it to characterize the sequences that probabilistic automata can predict.

\subsection{Probabilistic finite betting automata}

A probabilistic finite betting automaton is given by a tuple $(\S, \alph, \delta, s_0, \betfunc)$ where $\S$ is a finite set of states, $\alph$ is our alphabet, $\delta: \S \times \alph \to \mathscr{D}(\S)$ is a stochastic transition function that determines the probability distribution over the next state given the current one and a symbol from the alphabet, $s_0 \in \S$ is the initial state and $\betfunc : \S \times \alph \to \R_{\geq 0}$
is a betting function that assigns a bet to every next symbol given the current state. To guarantee that the betting is fair, we impose the condition
\begin{align*}
\sum_{a \in \alph} \betfunc(s, a) = |\alph|
\end{align*}
for every $s \in \S$. From now on, we fix some probabilistic betting automaton for the rest of the section to simplify the notation.

Given some finite word $w \in \alph^*$, the state of the automaton after reading $w$ is a random variable that we will denote as $S_w$. We define these variables inductively 
taking as base case $S_\Lambda = s_0$ and then, for $w \in \alph^*$ and $a\in \alph$, the event $S_{wa} = s$ is given by
\[
\bigcup_{s' \in \S} \left[ S_w = s' \wedge \delta_{|w|+1}(s',a) = s \right], 
\]
where $\{\delta_n(s', b) : n\in \N, s' \in \S, b\in \alph\}$ is a family of independent random variables representing the transitions at each step such that the distribution of each $\delta_n(s', b)$ is the same
as $\delta(s', b)$. We will be interested in runs of the automaton against infinite sequences of the form $X \in \alph^\omega$. Then, for ease of notation, we will write $S^X_n$ to refer to the variable $S_{X[1\ldots n]}$.

The capital of the automaton at step $n$ when executed against an infinite sequence $X$ is described by a random variable $C_n^X$ defined inductively as
\begin{align*}
    &C_0^X = 1\\
    &C_{n+1}^X = C_n^X \betfunc(S_n^X, X_{n+1}). 
\end{align*}
We say that the automaton $\mathbb{E}$-succeeds against an infinite sequence~$X$ if 
\[
\limsup_{n\to \infty} \mathbb{E}[C_n^X] = \infty. 
\]
\begin{example}
    \normalfont See Figure~\ref{fig:automata_example} for an example of a probabilistic betting automaton over alphabet $\alph = \{a, b\}$. Given the infinite sequence $X = b b \ldots$ it holds that $S_0^X = s_0$, $S_1^X$ equals $s_2$ with probability $1/2$ and $s_3$ with the same probability, and finally
    $S_2^X$ is distributed in the same way as $S_1^X$. The expected capital at the second step is $\mathbb{E}[C_2^X] = 1.5$. A simple way to compute this is to observe that by the second step there are two ``paths'' that the automaton could have taken: $s_0 s_2 s_2$ and $s_0 s_3 s_3$. Both have 
    the same probability, and they both achieve a capital equal to $1.5$. 
\end{example}

\begin{figure}[ht]
    \centering
    \resizebox{0.7\textwidth}{!}{%
    \begin{tikzpicture}[>=stealth,thick,scale=0.7]

    \tikzset{
        state/.style={draw,circle,minimum size=1.8cm}
    }

    \node[state, align=center] (q0) at (-3,4) {$s_0$\\$1$};
    \node[state, align=center] (q1) at (5,4) {$s_1$\\$1.5$};
    \node[state, align=center] (q2) at (-3,-2) {$s_2$\\$0.5$};
    \node[state, align=center] (q3) at (5,-2) {$s_3$\\$0.5$};

    \draw[->] ([xshift=-1.2cm]q0.west) -- (q0.west);

    \draw[->] (q0) to[bend left=20] node[above]{$a$, 1} (q1);

    \draw[->] (q0) -- node[left]{$b$, 1/2} (q2);

    \draw[->] (q2) -- node[below]{$a$, 1/2} (q3);

    \draw[->] (q1) to[bend left=13] node[right]{$b$, 1} (q3);
    \draw[->] (q3) to[bend left=13] node[left]{$a$, 1/2} (q1);

    \draw[->] (q0) to[bend left=6]
        node[pos=0.2,right,xshift=2mm]{$b$, 1/2} (q3);
    \draw[->] (q3) to[bend left=6]
        node[pos=0.2,left,xshift=-2mm]{$a$, 1/2} (q0);

    \draw[->] (q1) edge[loop right] node{$a$, 1} ();
    \draw[->] (q2) edge[loop left]
        node[align=center]{$b$, 1\\ $a$, 1/2} ();
    \draw[->] (q3) edge[loop right] node{$b$, 1} ();

\end{tikzpicture}
    }
    \caption{Example of a probabilistic betting automaton over alphabet $\alph = \{a, b\}$ with initial state $s_0$. Each edge $s_i \overset{c, p}{\rightarrow} s_j$ indicates that
     $\delta(s_i, c)(s_j) = p$ (if there is no edge with label $c \in \alph$ between two nodes then the transition probability is 0). Inside each node $s$ we write the betting factor $\betfunc(s, a)$, which is bounded between~$0$ and~$2$ because of the fairness
     condition. The other betting factor $\betfunc(s, b)$ is defined implicitly as $2-\betfunc(s, a)$.}
     \label{fig:automata_example}
\end{figure}
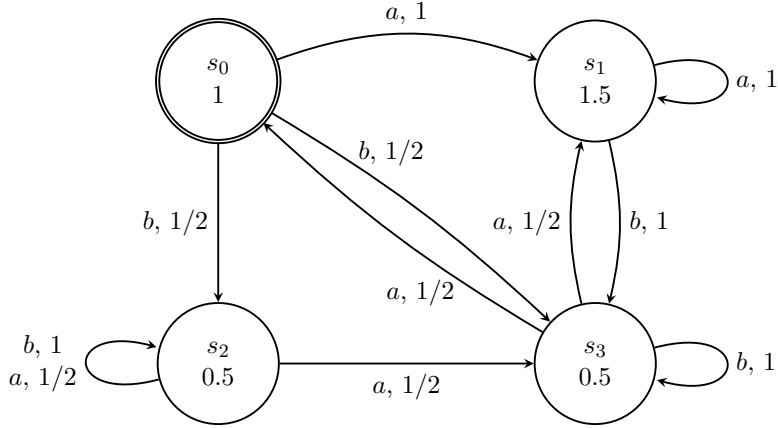

Our result regarding the notion of $\mathbb{E}$-success is the following.

\begin{theorem}\label{teo:prob_automata_fail_against_normal}
    Let $X\in \alph^\omega$ be a normal sequence. Then, it holds that $\lim_{n\to \infty} \mathbb{E}[C_n^X]$ exists, is finite, and the convergence to this value is exponentially fast.
\end{theorem}

\begin{proof}
    We will define a vector $\v$ and a fair family of matrices $\{\M_a\}_{a \in \alph}$ such that
    \begin{align}\label{eq:capital_as_product}
        \mathbb{E}[C_n^X] = \norm{\v \M_{X[1\ldots n]}}. 
    \end{align}
    Then, the result follows from Theorem~\ref{thm:main}.

    Identifying $\S$ with $\interval$, consider the vector $\v = \e_{s_0}$ and let
    \begin{align*}
        \M_a = \D_a\T_a, 
    \end{align*}
    where $\D_a$ is the diagonal matrix given by $\D_a(s,s) = \betfunc(s, a)$ and $\T_a$ is the transition matrix related to symbol $a \in \alph$ given by $\T_a(s_1,s_2) = \delta(s_1, a)(s_2)$. Observe that this family of matrices is fair:
    \begin{align*}
        \sum_{a\in \alph} \norm{\v \M_a} &= \sum_{a \in \alph} \sum_{s,s' \in \S} \v(s) \betfunc(s, a) \delta(s, a)(s')\\
        &= \sum_{a \in \alph} \sum_{s \in \S} \v(s) \betfunc(s, a) \underbrace{\sum_{s' \in \S} \delta(s, a)(s')}_{=1}\\
        &= \sum_{s \in \S} \v(s) \underbrace{\sum_{a \in \alph} \betfunc(s, a)}_{= |\alph|}\\
        &= |\alph| \sum_{s \in \S} \v(s) = |\alph|\, \, \norm{\v}. 
    \end{align*}
    Let $\xi_n^X = \v \M_{X[1\ldots n]}$. We now prove by induction that
    \begin{align}\label{eq:xi_equals_expected_capital}
        \xi_n^X(s) = \mathbb{E}[C_n^X | S_n^X = s] \P\left( S_n^X = s \right). 
    \end{align}
  (when $\P\left( S_n^X = s \right)=0$, we interpret the right-hand-side as being~$0$).   This implies Eq.~\eqref{eq:capital_as_product} by the law of total probability for expectations.

    The case $n=0$ holds. For the inductive case, we note that
    \begin{align*}
        \xi_{n+1}^X(s) &= \sum_{s' \in \S} \xi_n^X(s') \betfunc(s', X_{n+1}) \delta(s', X_{n+1})(s)  \\
        &= \sum_{s' \in \S} \mathbb{E}[C_n^X | S_n^X = s'] \P\left( S_n^X = s' \right)\betfunc(s', X_{n+1}) \delta(s', X_{n+1})(s) \\
        &= \sum_{s' \in \S} \mathbb{E}[\underbrace{C_n^X\betfunc(S_n^X, X_{n+1})}_{C_{n+1}^X} | S_n^X = s'] \underbrace{\P\left( S_n^X = s' \right) \delta(s', X_{n+1})(s)}_{\P\left( S_n^X = s', S_{n+1}^X = s \right)} \\
        &= \sum_{s' \in \S} \mathbb{E}[C_{n+1}^X | S_n^X = s'] \underbrace{\P\left( S_n^X = s', S_{n+1}^X = s \right)}_{\P\left( S_{n}^X = s' | S_{n+1}^X = s \right) \P\left( S_{n+1}^X = s \right)} \\
        &= \P\left( S_{n+1}^X = s\right)\sum_{s' \in \S} \mathbb{E}[C_{n+1}^X | S_n^X = s'] \P\left( S_n^X = s' | S_{n+1}^X = s \right) \label{eq:include_s_in_expectancy}\\
        &= \P\left( S_{n+1}^X = s\right) \mathbb{E}[C_{n+1}^X | S_{n+1}^X = s],   
    \end{align*}
    which proves Eq.~\eqref{eq:xi_equals_expected_capital}. (The last equality uses the conditional independence of $C_{n+1}^X$ and $S_{n+1}^X$ given $S_n^X$, since the next transition is sampled using fresh randomness).
\end{proof}

As an immediate corollary we conclude the following.

\begin{corollary}\label{coro:expected_success_and_normality}
    Given a sequence $X \in \alph^\omega$, the following statements are equivalent:
    \begin{enumerate}
        \item\label{enum:first} $X$ is normal.
        \item\label{enum:second} No probabilistic betting automaton $\mathbb{E}$-succeeds against $X$.
    \end{enumerate}
\end{corollary}

\begin{proof}
    The direction \ref{enum:first}) $\implies$ \ref{enum:second}) follows from Theorem~\ref{teo:prob_automata_fail_against_normal}. The other direction is immediate from the fact that 
    if $X$ is not normal then there is some deterministic betting automaton that succeeds against it \cite{schnorr1972endliche}.
\end{proof}

This gives a new characterization of normality in terms of a family of automata that generalizes deterministic betting automata.

\subsection{Overview of the proof}

We end this section with a high-level description of our proof strategy for Theorem~\ref{thm:main}. We first work under $(\bigstar)$, which means that every coordinate can reach every other through a matrix product; the general case will follow by induction on strongly connected components. Let us first review the core ideas of the Schnorr-Stimm theorem. Imagine we have a very simple martingale, described by a 1-state automaton. This state has a betting function~$\betfunc$ such that $\sum_{a \in \alph} \betfunc(a) = |\alph|$. If this automaton bets against an infinite sequence $X$ the capital of the martingale at stage~$n$ is equal to 
\[
\prod_{a \in \alph} \betfunc(a)^{\mathrm{occ}(X,a,n)}, 
\]
where $\mathrm{occ}(X,a,n)$ is the number of occurrences of the letter~$a$ in the first~$n$ values of $X$. When $X$ is normal, we have in particular that all letters are equally distributed, that is $\mathrm{occ}(X,a,n) = n/|\alph| + o(n)$, which thus means a capital of 
\[
\prod_{a \in \alph} \betfunc(a)^{n/|\alph|+o(n)}
\]
for the martingale. The natural way to study this quantity is to take its logarithm, which is 
\[
(n+o(1)) \cdot \frac{1}{|\alph|} \sum_{a \in \alph} \log(\betfunc(a)), 
\]
and use the strict concavity of the $\log$ to apply Jensen's inequality, from which we obtain that 
\[
\frac{1}{|\alph|} \sum_{a \in \alph} \log(\betfunc(a)) \leq \log \left( \frac{\sum_{a \in \alph} \betfunc(a)}{|\alph|} \right) = 0,
\]
the inequality being strict when not all $\betfunc(a)$ are equal. Hence, either all $\betfunc(a)$ are equal -- meaning that the martingale never bets anything -- and the capital then remains constant or not all $\betfunc(a)$ are equal and the capital decreases exponentially fast. 

The full Schnorr-Stimm theorem, when there is more than one state in the automaton describing the martingale, is obtained by a fairly simple extension of this analysis. Indeed, one can argue that if~$X$ is a normal sequence and some martingale playing against~$X$ can be computed by a finite automaton, then the only states visited infinitely often are in a strongly connected component of the automaton and must each be visited with some positive frequency. If all of these states are non-betting, the capital remains constant, but if at least one of them is betting, one can use the above analysis to show that the betting states incur an exponential loss. \\

In our proof of Theorem~\ref{thm:main}, what serves as the analogue of capital at stage~$n$ is the norm of the current vector $\v \M_{X[1\ldots n]}$. And like in the Schnorr-Stimm theorem, it will be advantageous to consider its logarithm to gauge at each step how much risk is taken in proportion to the capital. This risk will depend on the \emph{direction} $(\v \M_{X[1\ldots n]} / \norm{\v \M_{X[1\ldots n]}})$ which, unlike in the Schnorr-Stimm theorem, ranges over an infinite number of possible values. 

The key idea is that this direction has some form of ``forgetfulness'': if we know a sufficiently long suffix of $X[1\ldots n]$, then we can infer the direction $(\v \M_{X[1\ldots n]} / \norm{\v \M_{X[1\ldots n]}})$ with good precision. In the particular case where all matrices $\M_a$ are positive, one would obtain this fact using Hilbert's pseudo-distance between vectors -- which quantifies how far from collinear they are -- and Birkhoff's theorem which asserts that linear transformations arising from positive matrices contract the Hilbert pseudo-distance. However, our matrices are only assumed to be non-negative, so the vectors $\v \M_{X[1\ldots n]}$ may not have full support and the Hilbert pseudo-distance between two vectors can only be defined if they share the exact same support. This is why we first need to study, even before the trajectory of the direction $(\v \M_{X[1\ldots n]} / \norm{\v \M_{X[1\ldots n]}})$, the trajectory of the support $\supp(\v \M_{X[1\ldots n]})$. Luckily, the support is a finite object and this trajectory can be computed by a finite automaton. A fine study of this automaton allows us to define, and show the existence of, words $w$ which guarantee a Birkhoff-like behavior of the matrix $\M_w$, which we call \emph{pseudo-mixing} words. 

Since all finite words appear in a normal sequence with some positive frequency, the existence of pseudo-mixing words will guarantee the desired forgetfulness of the trajectory $(\v \M_{X[1\ldots n]} / \norm{\v \M_{X[1\ldots n]}})$ when~$X$ is a normal sequence. It is at this point that we'll encounter the central case disjunction of the proof:
\begin{itemize}
\item Case 0 will be the situation where the matrices $\M_a$ are arranged in a way that $\v \M_{X[1\ldots n]}$ is guaranteed to become the null vector at some stage.
\item Case 1 will be the situation where some risk-taking direction $\x$ is an accumulation point of $(\v \M_{X[1\ldots n]} / \norm{\v \M_{X[1\ldots n]}})$. In this case, thanks to the forgetfulness, it will be possible to guarantee that the martingale is taking risks and use a concentration inequality, namely Azuma's inequality, to prove the exponential loss of our martingale. 
\item Finally, Case 2 will correspond to the situation where there exists a non-betting trajectory. By forgetfulness, we will then argue that every trajectory will get exponentially close to the non-betting one and thus its ``capital''  \norm{\v \M_{X[1\ldots n]}} will stabilize exponentially fast (or in certain cases will hit the $\zero$ vector).  
\end{itemize} 

In all of these cases, our main theorem will follow.

\section{Some elements of probability theory and linear algebra}

Let us now gather the mathematical tools we will need to prove our main theorem.

\subsection{The Hilbert projective metric and Birkhoff's theorem}

To understand the behaviour of $(\v \M_{X[1\ldots n]})_{n \in \mathbb{N}}$ we will identify recurring \textit{directions} of these vectors, which correspond to accumulation points of
$(\v \M_{X[1\ldots n]} / \norm{\v \M_{X[1\ldots n]}})_{n \in \mathbb{N}}$ (which is well-defined whenever the original sequence is not ultimately 0). We will use the Hilbert metric as a tool to measure distance between
these directions.  

\begin{definition}
The Hilbert projective metric $\dh$ is defined on pairs of non-negative vectors having the same non-empty support. Let $\u, \v \in \vect$ be two non-negative vectors that have the same non-empty support~$E$. We set 
\begin{align*}
\dh(\u,\v) = \ln \left( \frac{\max_{i \in E} (\u(i)/\v(i))}{\min_{i \in E} (\u(i)/\v(i))} \right) = \ln \left( \max_{i \in E} (\u(i)/\v(i)) \right) + \ln \left( \max_{i\in E} (\v(i)/\u(i)) \right). 
\end{align*}

\end{definition}

One can readily verify that $\dh$ is a pseudo-metric: If $\u, \u', \u'' \in \pvect$ have the same support, we have $\dh(\u,\u)=0$, $\dh(\u,\u')=\dh(\u',\u)$ and $\dh(\u,\u'') \leq \dh(\u,\u')+\dh(\u',\u'')$. However, $\dh(\u,\u')=0$ does not imply $\u=\u'$ but only implies that $\u$ and $\u'$ are collinear (and indeed $\dh$ is invariant under scalar multiplication: if $\lambda >0$, $\dh(\lambda \u,\v) = \dh(\u, \v)$). As a result, $\dh$ induces, for every non-empty~$E$, an actual metric on the set $\U_m^+(E)$ of non-negative vectors of norm~$1$ and support~$E$. This metric is not linearly equivalent to the metric induced by the norm (for example, for $\varepsilon$ small, if we take $\u=(\varepsilon,1-\varepsilon-\varepsilon^2,\varepsilon^2)$ and $\v=(\varepsilon^2,1-\varepsilon-\varepsilon^2,\varepsilon)$,  $\norm{\u-\v}=2(\varepsilon-\varepsilon^2)$ is small and $\dh(\u,\v)=2\ln(1/\varepsilon)$ is large) but it induces the same topology on $\U_m^+(E)$, and we do have at least a one-sided inequality.

\begin{lemma}\label{lem:dh-vs-norm}
If $\u,\v \in \U_m$ have the same support, then $\norm{\u-\v} \leq \dh(\u,\v)$.
\end{lemma}

\begin{proof}
Let $E$ be the support of $\u$ and $\v$. Let $a =  \max_{i \in E} (\u(i)/\v(i))$ and $b=  \max_{i \in E} (\v(i)/\u(i)) $ so that $\dh(\u,\v)=\ln a+ \ln b$. Since $\u$ and $\v$ have norm~$1$, both $\max_{i \in E} (\u(i)/\v(i))$ and $\max_{i \in E} (\v(i)/\u(i))$ must be at least~$1$, i.e., $a \geq 1$ and $b \geq 1$. We also have

\begin{align*}
\norm{\u-\v}  & = \sum_{i \in E} |\u(i) -\v(i)| \\
 & = \sum_{i  \in E\, : \, \u(i) \geq \v(i)} \u(i) \big(1-\v(i)/\u(i)\big) + \sum_{i  \in E\, : \, \v(i) > \u(i)} \v(i) \big(1-\u(i)/\v(i)\big)\\
 &  \leq \sum_{i \in E\, : \, \u(i) \geq \v(i)} \u(i) (1-1/a) + \sum_{i \in E\, : \, \v(i) > \u(i)} \v(i) (1-1/b)\\
 & \leq (1-1/a) + (1-1/b). 
 \end{align*}

Using the fact that $1-1/x \leq \ln(x)$ for all~$x \geq 1$, it follows that 
\[
\norm{\u-\v} \leq \ln(a) + \ln(b) = \dh(\u,\v). \qedhere
\]
\end{proof}

We also show that the Hilbert distance is not increased by non-negative linear operators.

\begin{lemma}\label{lem:dh-perservation}
Let $\u,\v \in \vect$ having common nonempty support and $\M \in \matm$. Then, $\u \M$ and $\v \M$ have the same support, and if this support is non-empty, 
\[
\dh(\u \M, \v \M) \leq \dh(\u, \v).
\]
\end{lemma}

\begin{proof}
Let $E$ be the support of $\u$ and $\v$ and again let $a =  \max_{i \in E} (\u(i)/\v(i))$ and $b=  \max_{i \in E} (\v(i)/\u(i))$ so that $\dh(\u,\v)=\ln a+ \ln b$. Let $\u'= \u \M$ and $\v' = \v\M$. For all $j \in \interval$, we have 
\[
\u'(j) = \sum_{i=1}^m \u(i)\M(i,j) \geq \sum_{i=1}^m (1/b)\v(i)\M(i,j) = \v'(j) /b, 
\]
and symmetrically, $\v'(j) \geq \u'(j)/a$. This both shows that $\u'$ and $\v'$ have the same support and, if the support $E'$ of $\u'$ and $\v'$ is non-empty, that $a' =  \max_{i \in E'} (\u'(i)/\v'(i)) \leq a$ and $b'=  \max_{i \in E'} (\v'(i)/\u'(i)) \leq b$. Therefore,
\[
\dh(\u',\v') = \ln(a')+\ln(b') \leq \ln(a)+ \ln(b) = \dh(\u,\v). \qedhere 
\]
\end{proof}

Finally, when the linear operator is strictly positive we can guarantee contraction using Birkhoff's Theorem.

\begin{theorem}[Birkhoff's Theorem \cite{birkhoff}]
Let $r,s$ be positive integers. For any $\M \in \mat{r}{s}{\mathbb{R}_{>0}}$, there exists a constant $\tau(\M) \in (0,1)$ such that for any two $\v,\v' \in \mathbb{R}^r_{>0}$, 
\[
\dh(\v \M,\v' \M) \leq \tau(\M)\cdot \dh(\v,\v').
\]
\end{theorem}

\subsection{Martingales}

Martingales are a central concept of probability theory inspired by the concept of betting strategy but whose definition is very general. In this paper, we follow the traditional terminology of algorithmic randomness and call \emph{martingale} a function
\[
D : \alph^* \rightarrow \mathbb{R}
\]
such that for every $w \in \alph^*$, 
\[
\frac{1}{|\alph|} \sum_{a \in \alph} D(wa) = D(w).
\]
When we only have 
\[
\frac{1}{|\alph|} \sum_{a \in \alph} D(wa) \leq D(w), 
\]
we say that $D$ is a supermartingale.

This terminology is only a slight abuse: If $D$ is a martingale (resp. supermartingale) in our sense and $X \in \alph^\omega$ is an infinite sequence chosen uniformly at random, the sequence of random variables $(D(X[1 \ldots n]))_{n \in \N}$ is a martingale (resp. supermartingale) in the probability-theoretic sense. 

One should think of~$D$ as the capital of a player who bets on the values of the sequence and gets rewarded fairly in the case of martingales (this is the condition  $\sum_{a \in \alph} D(wa) / |\alph| = D(w)$). In the case of supermartingales, as in real life, the odds may be stacked against the gambler, hence the inequality $(1/|\alph|) \sum_{a \in \alph} D(wa) \leq D(w)$. 

For technical reasons, we will allow supermartingales to take the value $-\infty$.\\

The two main theorems we will need from the theory of martingales are two famous inequalities, respectively known as Kolmogorov's inequality and Azuma's inequality. 

\begin{theorem}[Kolmogorov's inequality] \label{thm:kolmogorov-inequality}
Let $D$ be a non-negative supermartingale with $D(\emptystr)\leq 1$. For every constant~$C>0$, if~$X \in \alph^\omega$ is chosen uniformly at random, 
\[
\P \big(\exists n\,  D(X[1 \ldots n]) \geq C \big) \leq 1/C. 
\]
In particular, for every~$n$, 
\[
\P \big(D(X[1 \ldots n]) \geq C \big) \leq 1/C. 
\]
\end{theorem}

See for example~\cite[Theorem 6.3.3]{DowneyH2010} for a proof. 

\begin{theorem}[Azuma's inequality] \label{thm:azuma-inequality}
Let $D: \alph^* \rightarrow (-\infty,\infty)$ be a supermartingale and $(c_k)_{k \in \N}$ be a sequence of positive reals such that,  $|D(wa) - D(w)| \leq c_{|w|+1}$ for all $w \in \alph^*$ and $a \in \alph$. Then, for every~$n \geq 1$ and $\theta>0$, if $w$ is a string of length~$n$ picked uniformly at random,
\[
\P_{|w|=n}(D(w) - D(\emptystr) > \theta  ) \leq \exp \left( \frac{-\theta^2}{2 \sum_{k=1}^n c_k^2} \right).  
\]
\end{theorem}

See~\cite{AlonS2016} for a proof. \\

While it may not be true for some martingales in the probability-theoretic sense, our martingales have a very simple form which allows us to weaken the hypothesis of bounded variations ($|D(wa) - D(w)| \leq c_{|w|+1}$ in our case) to bounded increases. We state this in the proposition below where we further take the family $(c_k)_{k \in \N}$ to be a constant, which is sufficient for our needs. 

\begin{proposition}\label{prop:subexponential_increase_of_A_martingales}
Let $D: \alph^* \rightarrow [-\infty,\infty)$ be a supermartingale with $D(\emptystr)>-\infty$ such that for some constant~$c>0$, $D(wa) \leq D(w)+c$ for any $w \in \alph^*$ and $a \in \alph$. Then there is a constant $d>0$ such that, for every $n\geq 1$ and $\theta>0$, if $w$ is a string of length~$n$ picked uniformly at random in~$\alph^n$,
\[
\P_{|w|=n}(D(w) - D(\emptystr) > \theta) \leq \exp \left( \frac{-\theta^2}{dn} \right).
\]
\end{proposition}

\begin{proof}
Let~$D$ be such a supermartingale. To apply Azuma's inequality, we need a lower bound for the variation. Consider the real-valued function~$\tilde{D}$ defined by $\tilde{D}(\emptystr)=D(\emptystr)$ and inductively, if $\tilde{D}(w)$ is already defined and $D(w)>-\infty$, set for every~$a \in \alph$:
\[
\tilde{D}(wa) = \tilde{D}(w) + \max \big( D(wa) - D(w), -c |\alph| \big).
\]
If $D(w)=-\infty$, set $\tilde{D}(wa)=\tilde{D}(w)$ for every~$a \in \alph$.
Since $D(w)=-\infty$ implies $D(wa)=-\infty$ for all~$a$, and $\tilde{D}(wa) - \tilde{D}(w) \geq D(wa) - D(w)$ whenever $D(w)>-\infty$, it follows by induction that $\tilde{D} \geq D$. We further claim that~$\tilde{D}$ is a supermartingale. Indeed, let $w \in \alph^*$. If $D(w)=-\infty$, the supermartingale inequality holds with equality by construction. Otherwise, we have
\[
\tilde{D}(w) - \frac{1}{|\alph|} \sum_{a \in \alph} \tilde{D}(wa) = - \frac{1}{|\alph|} \sum_{a \in \alph} \max \big( D(wa) - D(w), -c |\alph| \big).
\]
If $D(wa) - D(w) \geq -c|\alph|$ for every $a \in \alph$, then the right-hand side is equal to
\[
D(w) - \frac{1}{|\alph|} \sum_{a \in \alph} D(wa),
\]
which is nonnegative as~$D$ is a supermartingale. In the opposite case where $D(wb) - D(w) < -c|\alph|$ for some~$b \in \alph$, the right-hand side is equal to
\[
c - \frac{1}{|\alph|} \sum_{a\not=b} \max \big( D(wa) - D(w), -c |\alph| \big),
\]
and since $D(wa) - D(w) \leq c$ for all~$a$ by our hypothesis, it follows that the right-hand side is also non-negative in that case. This shows that $\tilde{D}$ is indeed a supermartingale.

Since our construction ensures that $\left| \tilde{D}(wa) - \tilde{D}(w) \right| \leq c|\alph|$ for all~$w$ and $a$, we can apply Azuma's inequality to~$\tilde{D}$ and get, for~$w$ of length~$n$ chosen at random
\[
\P(\tilde{D}(w)-\tilde{D}(\emptystr) > \theta) \leq \exp \left( \frac{-\theta^2}{2nc^2|\alph|^2} \right),
\]
and since $\tilde{D}$ is greater or equal to~$D$ and $\tilde{D}(\emptystr)=D(\emptystr)$, we have, a fortiori,
\[
\P(D(w)-D(\emptystr) > \theta) \leq \exp \left( \frac{-\theta^2}{2nc^2|\alph|^2} \right),
\]
hence we have proven the desired result with $d=2c^2|\alph|^2$.
\end{proof}

Martingales and supermartingales are particularly well adapted to the theorem we want to prove. Indeed, it follows directly from the superfairness of our family $\{\M_a\}_{a \in \alph}$ of matrices that given a vector $\v$, the function $w \mapsto \norm{\v \M_w}$ is a supermartingale (and a martingale if the family $\{\M_a\}_{a \in \alph}$ is fair). 

We further consider its logarithm which we denote by $L^\v$:
\[
L^\v(w) = \ln \norm{\v \M_w}
\]

(where we use the convention $\ln(0)=-\infty$ when $\v \M_w=\zero$). 
By concavity of the logarithm and Jensen's inequality, we have 
\begin{equation}\label{eq:jensen}
\ln \norm{\v \M_w} \geq \ln \left( \frac{1}{| \alph |} \sum_{a \in \alph} \norm{\v \M_{wa}} \right) \geq \frac{1}{| \alph |} \sum_{a \in \alph} \ln \norm{\v \M_{wa}},
\end{equation}
(the first inequality is due to the superfairness of the family $(\M_a)_{a \in \alph}$). In other words, we have proven the following important lemma. 

\begin{lemma}
For any $\v$, $L^\v$ is a supermartingale (which can take the value $-\infty$).
\end{lemma}

\noindent In the inequality~(\ref{eq:jensen}), the difference
\[
\ln \norm{\v \M_w} - \frac{1}{| \alph |} \sum_{a \in \alph} \ln \norm{\v \M_{wa}}
\]
(which we have just proven to be nonnegative) can be seen as a measure of the risk taken by the martingale $L^\v$, having already seen~$w$, on its bet on the next letter. We would like, in the spirit of Doob's decomposition, to write $L^\v$ as the difference between a martingale and the sum of the risks taken. However, $L^\v$ may take the value $-\infty$, and we want to avoid dealing with quantities such as $(-\infty)+(+\infty)$. This is why we will put an artificial bound~$1$ on this risk measure, while keeping the fact that the sum of $L$ and the cumulative risk is a supermartingale. 

\begin{definition}
For $\u \in \vect \setminus \{\zero\}$ we define a bounded risk measure as
\[
 \Delta(\u) =  \min\left(1, \ln \norm{\u} - \frac{1}{| \alph |} \sum_{a \in \alph} \ln \norm{\u \M_{a}}\right), 
 \]

where, as above, we take $\ln(0) = -\infty$. By convention, we also set $\Delta(\zero)=0$. 
\end{definition}

\noindent A crucial remark is that $\Delta$ is invariant under positive scalar multiplication and thus depends only on the direction $\frac{\u}{\norm{\u}}$ of the vector $\u$.\\

The cumulative risk~$\Gamma$ is defined as follows:

\begin{definition}\label{def:Gamma}
For $\v \in \vect$ and $w \in \alph^*$, we set:
\[
\Gamma^\v(w) = \sum_{w' \sqsubset w} \Delta(\v \M_{w'}). 
\]
\end{definition}

As announced, this definition is enough to give us the following. 

\begin{lemma}
For any $\v \in \vect$, $L^\v + \Gamma^\v$ is a supermartingale.
\end{lemma}

\begin{proof}
For $w \in \alph^*$, using the fact that, by definition, $\Gamma^\v(wa)=\Gamma^\v(w)+\Delta(\v \M_w)$, we get
\begin{equation}\label{eq:l-gamma}
\frac{1}{|\alph|} \sum_{a \in \alph} (L^\v(wa) + \Gamma^\v(wa))    =  \frac{1}{|\alph|} \sum_{a \in \alph} \ln \norm{\v \M_{wa}} +  \Big(\Gamma^\v(w) + \Delta(\v \M_w) \Big).\\ 
\end{equation}
-- If $\norm{\v \M_{wa}} >0$ for all~$a$, then using the definition of~$\Delta$, we get:
\begin{eqnarray*}
\frac{1}{|\alph|} \sum_{a \in \alph} (L^\v(wa) + \Gamma^\v(wa))  &  \leq & \frac{1}{|\alph|} \sum_{a \in \alph} \ln \norm{\v \M_{wa}} +  \Gamma^\v(w) +  \ln \norm{\v\M_w} - \frac{1}{| \alph |} \sum_{a \in \alph} \ln \norm{\v \M_{wa}} \\ 
 & = & L^\v(w) + \Gamma^\v(w), 
 \end{eqnarray*}
as desired. \\

\noindent -- And if $\norm{\v \M_{wa}} =0$ for some~$a$, the right-hand side of~(\ref{eq:l-gamma}) becomes $-\infty$ and therefore 
\[
\frac{1}{|\alph|} \sum_{a \in \alph} (L^\v(wa) + \Gamma^\v(wa)) \leq L^\v(w) + \Gamma^\v(w)
\]
holds trivially.

\end{proof}

\subsection{Markov chains and deterministic finite automata}

\newcommand{\aMarkovChain}{\mathcal{Z}}

We will need a basic result regarding irreducible Markov chains: namely, the fact that for any starting distribution it holds that with probability one each state will be visited infinitely often with some fixed frequency depending only on the chain itself.

\begin{lemma}\label{lemma:markov_ergodic_theorem}
	Let $(\aMarkovChain_n)_{n \in \mathbb{N}}$ be an irreducible Markov chain over states $\S$ with stationary distribution $\pi$. Then, for any starting distribution $\mu$ and $s\in \S$,
	\begin{align*}
		\P_{\mu}\left( \lim_{n \to \infty} \frac{|\{1 \leq i \leq n: \aMarkovChain_i = s\}|}{n} = \pi(s)\right) = 1. 
		\end{align*}
\end{lemma}

See~\cite[Theorem C.1]{levin2017markov} for a proof. In particular, we will use the following proposition which relates Markov chains to deterministic finite automata:

\begin{proposition}\label{prop:automata_behaves_as_markov_chain}
	Let $\A = (\S, \alph, \delta, s_0)$ be a strongly connected deterministic finite automaton over alphabet $\alph$ with states $\S$, initial state $s_0$ and transition function $\delta:\S\times \alph \to \S$. Assuming that we sample a sequence $Z \in \alph^\omega$ uniformly at random, consider the random variables $\aMarkovChain_n^l: \alph^\omega \to \S \times \alph^l$ given by
	\begin{align*}
		\aMarkovChain_n^l(Z) = (\delta^*(s_0, Z[1\ldots n]), Z[n+1\ldots n+l]), 
	\end{align*}
	where $\delta^*$ is the extended transition function for $\A$. Then, $(\aMarkovChain_n^l)_{n \in \mathbb{N}}$ is a Markov chain for every $l$ satisfying
	\begin{align*}
		\P_{s_0}\left( \lim_{n\to \infty} \frac{|\{1 \leq i \leq n: \aMarkovChain_i^l = (s, z)\}|}{n}  = \pi(s)|\alph|^{-l} \right) = 1
	\end{align*}
	for every $s \in \S$ and $z \in \alph^l$ where $\pi$ is the stationary distribution of $(\aMarkovChain_n^0)_{n \in \N}$.
\end{proposition}

\begin{proof}
	This result is proven in \cite[Theorem 4.1]{schnorr1972endliche}. In a nutshell, it is shown as follows. Lemma~\ref{lemma:markov_ergodic_theorem} implies the result for the Markov chain $(\aMarkovChain_n^0)_{n \in \mathbb{N}}$. For any other chain $(\aMarkovChain_n^l)_{n \in \mathbb{N}}$ the result holds by observing that for any index~$n$ the value of the next~$l$ characters from~$Z$ is independent of the state of the automaton at step~$n$.
\end{proof}

This proposition states that, when facing a random sequence, each state of the automaton will be visited with some positive frequency, and, moreover, the distribution of the next~$l$ characters of the sequence is completely independent of the current state. \\

As is standard in graph theory, we say that a subset $\S'$ of the set of states is a \emph{strongly connected component (SCC)} if for any two $s,s' \in \S'$ there is a path from $s$ to~$s'$ and it is maximal (with respect to inclusion) with this property. We say that an SCC is a \emph{bottom} one (or BSCC) if it has no outgoing edge to other SCC's and it is a \emph{top} one (or TSCC) if there is no incoming edge to it from other SCC's. 

We will need a standard fact regarding the existence of ``BSCC-synchronizing words'' that we will use in later proofs without referring to it explicitly.

\begin{lemma}\label{lemma:BSCC_sync}
	Let $\A = (\S, \alph, \delta, s_0)$ be a deterministic finite automaton. Let $\B_1,\ldots, \B_k$ be all the BSCCs of $\A$. Then, there is a word (which we call a ``BSCC-synchronizing'' word) $w_{sync}$ satisfying $\delta^*(s, w_{sync}) \in \bigcup_{i=1}^k \B_i$ for every $s \in \S$.
\end{lemma}

\begin{proof}
	Let $\S = \{s_1, \ldots, s_n\}$ and $\B = \bigcup_{i=1}^k \B_i$. We build the word $w_{sync}$ iteratively in $n$ steps, where at step $i$ we obtain a word $w_i$ such that $\delta^*(s_j, w_i) \in \B$ for every $1 \leq j \leq i$. We let $w_0 = \Lambda$, and to define
	the $(i+1)$th word we observe that by definition there must exist some word $z$ such that $\delta^*(s_{i+1}, w_{i}z) \in \B$. Therefore, we can take $w_{i+1} = w_{i}z$ and preserve the invariant. The result is obtained by taking $w_{sync} = w_n$.
\end{proof}

\section{Main result under the $(\bigstar)$ hypothesis}

Recall that $(\bigstar)$ asserts that for every $i,j\in\interval$ there is a word $w\in\alph^*$ with $\M_w(i,j)>0$. In this section we prove Theorem~\ref{thm:main} under this assumption, distinguishing three possible behaviours of the norm.

We will think of the characters $a \in \alph$ as acting on vectors through the corresponding matrix $\M_a$. More precisely, we will distinguish two different spaces on which these symbols act: first, on the set of 
unitary vectors $\v \in \U_m$ as an evolution $\v \cdot a = \frac{\v \M_{a}}{\norm{\v\M_{a}}}$, and second, over sets $E$ from $2^{\interval}$ by just considering the support of vectors. These actions will be naturally extended to $\alph^*$.

\subsection{Action of $\alph^*$ on $\U_m \cup \{\zero\}$}

The semigroup $\alph^*$ acts on $\U_m \cup \{\zero\}$ via the following operation:
\[
\v \cdot w = \left\{ \begin{array}{ll} \zero & ~ \text{if} ~ \v \M_w = \zero \\ \frac{\v \M_w}{\norm{\v \M_w}}  &  ~\text{otherwise}\end{array} \right.. 
\]

Since the pseudo-distance $\dh$ is invariant under positive scalar multiplication of any of its components, it follows that for $\v, \v' \in \U^+_m$, if $\v \M_w$ and $\v' \M_w$ are non-zero, then
\[
\dh(\v \cdot w ,\v' \cdot w) = \dh(\v \M_w , \v' \M_w). 
\]

The action is continuous when restricted appropriately.

\begin{lemma}\label{lemma:continuity_of_action}
	Let $E \subseteq [1\ldots m]$ and assume that for every $\u \in \U_m^+(E)$ and $w \in \alph^*$ it holds that $\u \cdot w \neq \zero$. Then, the functions $\Delta:\R^m_{\geq 0}(E) \to \mathbb{R}$ and $\cdot \,w:\U_m^+(E) \to \U_m$ defined previously are continuous with respect to either $\norm{\cdot}$ or $d_H$ for every $w \in \alph^*$. 
\end{lemma}

\begin{proof}
	Under the hypothesis, if $\u \in \R^m_{\geq 0}(E)$ then $\u \M_a \neq \zero$ for all $a \in \alph$, which implies the continuity of $\Delta$ thanks to the continuity of $\ln$ in $(0,\infty)$. The continuity of $\cdot \,w$ is immediate from $\u \M_w \neq \zero$.
\end{proof}

\subsection{Action of $\alph^*$ on $2^\interval$ and pseudo-mixing words}

The semigroup $\alph^*$ also acts naturally on the power set $\ps$ via the following operation (also denoted $\cdot$ for simplicity):

\[
E \cdot w = \supp(\v_E \M_w), 
\]
where $\v_E$ is any vector of support $E$ (one can readily verify that the definition is independent of the particular choice of $\v_E$). Note that for any $\v \in \U_m$ and $w \in \alph^*$ we have $\supp(\v \cdot w) = \supp(\v) \cdot w$. Note also that this action is compatible with the union of sets: $(E \cup E') \cdot w = (E \cdot w) \cup (E' \cdot w)$. As a corollary, it is also compatible with inclusion: if $E \subseteq E'$, then $E \cdot w \subseteq E' \cdot w$. \\

Let $E \subseteq \interval$. We say that $w \in \alph^*$ \emph{stabilizes} $E$ if $E \cdot w = E$. Furthermore, we say that $w$ \emph{mixes}~$E$ if for all $i \in E$, $\{i\} \cdot w = E$ (note that this in particular implies that $w$ stabilizes~$E$). This is equivalent to saying that $w$ stabilizes~$E$ and that the submatrix $\M_w^{[E \times E]}$ is positive. 
Lastly, we say that $w$ \emph{pseudo-mixes}~$E$ if $w$ stabilizes~$E$ and for every~$i \in E$, either $\{i\} \cdot w = E$ or $\{i\} \cdot w = \emptyset$. 
This is equivalent to saying that $w$ stabilizes~$E$ and the submatrix $\M_w^{[E\times E]}$ is equal to $\mathbf{P}^{E'}\M_{w}^{[E'\times E]}$ where $E'\subseteq E$ is the set of~$i$ such that $\{i\} \cdot w = E$, $\M_{w}^{[E'\times E]}$ is positive, and $\mathbf{P}^{E'}$ is the projection over the indices from $E'$ (i.e., $\mathbf{P}^{E'}$ is the map that ``erases'' the coordinates $i \notin E'$). 
\\

Now, consider the deterministic finite automaton $\A$ (without initial state) with set of states $2^\interval$ where we put a transition $E \overset{a}{\rightarrow} E'$ if $E \cdot a = E'$. We call its directed graph the \emph{support graph}. The empty set $\emptyset$ is stable under dot product: $\emptyset \cdot w = \emptyset$ for all words~$w$. Hence, $\{\emptyset\}$ is a BSCC of $\A$, which we call the null BSCC. If it is the only BSCC of $\A$, things are very simple, as we will see in Proposition~\ref{prop:case_0}. If not, under the assumption $(\bigstar)$ we show that there is a~BSCC $\B$ of $\A$ which has very useful properties for our purposes. 

For example, with $m=1$, $\alph=\{0,1\}$, $\M_0=(2)$ and $\M_1=(0)$, the support graph has states $\{1\}$ and $\emptyset$. The letter $0$ fixes both states, while $1$ sends $\{1\}$ to $\emptyset$, so the null BSCC is reachable from~$\{1\}$.

\begin{lemma} \label{lem:unique-bscc}
Under the assumption $(\bigstar)$, if there exists at least one non-null BSCC in $\A$, then there is a BSCC~$\B$ of $\A$ with the following properties.  

\begin{itemize}
\item[(a)] For any $i \in \interval$, there exists $E \in \B$ such that $i \in E$. 
\item[(b)] For any $i \in \interval$, $\B$ is reachable from~$\{i\}$, and is the unique non-null BSCC reachable from~$\{i\}$. 
\item[(c)] For any $E$ in $\B$, there is a $w \in \alph^*$ which pseudo-mixes $E$. 
\end{itemize}

\end{lemma}

\begin{proof}

Assuming there are non-null BSCC's, let us fix some non-empty $F \in \ps$ which is minimal with respect to inclusion among all members of all non-null BSCC's of $\A$. Let~$\B$ be the BSCC containing~$F$. We claim that $\B$ has the desired properties. \\

$(a)$ Let $i \in \interval$. Let $j $ be some element of $F$. By $(\bigstar)$, there is a $w \in \alph^*$ such that $i \in \{j\} \cdot w$ and hence $i \in F \cdot w$ (we saw that the dot-product preserves inclusion). But $F$ is in $\B$ which is a BSCC, hence $F \cdot w$ is in $\B$ as well. \\

$(b)$ Let $i \in \interval$. Let us first show that $\B$ is the only possible non-null BSCC reachable from~$\{i\}$. Let $\B'$ be any non-null BSCC reachable from~$\{i\}$, witnessed by a word~$w$ such that $\{i\} \cdot w \in \B'$. By item $(a)$, there exists some $E \in \B$ such that $i \in E$. Since $E$ is in the SCC of $F$, so is $E'=E \cdot w$, and thus let $w'$ be such that $E' \cdot w'=F$. Then, on the one hand, we have $\{i\} \cdot ww' = (\{i\} \cdot w) \cdot w'$ so $\{i\} \cdot ww' \in \B'$ because $\{i\} \cdot w$ belongs to $\B'$ which is a BSCC. On the other hand, since $\{i\} \subseteq E$, we have $\{i\} \cdot ww' \subseteq E \cdot ww' = E' \cdot w' = F$. Hence, $\{i\} \cdot ww'$ is in a non-null BSCC (that is, $\B'$) but is also contained in~$F$. By minimality of~$F$, we must have $\{i\} \cdot ww' = F$ and thus $\B = \B'$.

Next, let us show that $\B$ is reachable from \emph{some} singleton. Let~$x$ be a BSCC-synchronizing word in the sense of Lemma~\ref{lemma:BSCC_sync}. The set $F \cdot x$ is in $\B$ and hence there is some extension $y$ of $x$ such that $F \cdot y = F$. Note that the set of BSCC-synchronizing words is closed under extension, so $y$ is BSCC-synchronizing as well. 

Now, by preservation of union under the dot-product, we have 
\[
F = F \cdot y = \bigcup_{i \in F} \{i\} \cdot y.  
\]
Since $F$ is non-empty, for some $j \in F$, $\{j\} \cdot y$ is non-empty. Moreover, since $y$ is a BSCC-synchronizing word, $\{j\} \cdot y$ is in a BSCC, which must thus be non-null. But since $\{j\} \cdot y \subseteq F$, by minimality of $F$, we must have $\{j\} \cdot y = F$. Hence, $\B$ is reachable from~$\{j\}$. Note that since $j$ was taken arbitrarily among the elements $i \in F$ such that $\{i\} \cdot y \not= \emptyset$, this argument also shows that the word~$y$ pseudo-mixes~$F$. 

Now that we have established that $\{j\} \cdot y = F$ for some $j \in F$, which we fix, let us consider some other $i \in \interval$. By $(\bigstar)$, there is some $w$ such that $\{j\} \subseteq \{i\} \cdot w$. Hence, $F = \{j\} \cdot y \subseteq \{i\} \cdot wy$, so $\{i\} \cdot wy$ is non-empty. Since $y$ is BSCC-synchronizing, $\{i\} \cdot wy$ is in a BSCC. But we have seen that $\B$ is the only possible non-null BSCC reachable from singletons. Hence, $\{i\} \cdot wy \in \B$, as wanted. \\

$(c)$ We have already seen in the proof of $(b)$ that $y$ pseudo-mixes the set $F$. To show that other members of $\B$ have pseudo-mixing words, let $E \in \B$. By strong connectedness, there are $u, v \in \alph^*$ such that $E \cdot u = F$ and $F \cdot v = E$. We can verify that $uyv$ pseudo-mixes~$E$. Indeed, let $i \in E$. Since $\{i \} \subseteq E$, $\{i\} \cdot u \subseteq E \cdot u = F$. Since $y$ pseudo-mixes $F$, $\{i\} \cdot uy$ is either $F$ or is empty, and thus $\{i\} \cdot uyv$ is either $F \cdot v = E$ or is empty. \end{proof}

The interest of pseudo-mixing words resides in the following extension of Birkhoff's theorem.

\begin{proposition}\label{prop:birkhoff-pseudo-mixing}
Let $E \subseteq \interval$. If $w$ pseudo-mixes~$E$, there is a constant $\tau(E,w) \in (0,1)$ such that, if $\u$ and $\v$ have support~$E$, 
\[
\dh(\u \M_w, \v \M_w) \leq \tau(E,w) \dh(\u, \v). 
\]
Equivalently, if $\u, \v \in \U_m^+(E)$,
\[
\dh(\u \cdot w, \v  \cdot w) \leq \tau(E,w) \dh(\u, \v). 
\]
\end{proposition}

\begin{proof}
By definition, if $w$ pseudo-mixes~$E$ then for every~$i \in E$, either $\{i\} \cdot w = E$ or $\{i\} \cdot w = \emptyset$, with at least one index~$i$ of the first type. Let $E'$ be the set of indices of the first type and let~$r$ be its cardinality. As we said when we introduced pseudo-mixing words, the action of $\M_w$ on $\U_m^+(E)$ can be seen as the composition of the projection map~$\mathbf{P}^{E'}$ which erases the coordinates whose index is not in $E'$, followed by the linear map $\U_m^+(E') \rightarrow \U_m^+(E)$ induced by $\M_w$. Since the matrix $\M_w^{[E' \times E]}$ only has positive coefficients (this is the meaning of $\{i\} \cdot w = E$ for all~$i \in E'$), we have:
\begin{align*}
\dh(\u \cdot w, \v  \cdot w)& =  \dh \Big((\u \mathbf{P}^{E'})\M_w,(\v \mathbf{P}^{E'})\M_w\Big) \\
 &\leq   \tau(\M_w^{[E' \times E]})\, \dh(\u \mathbf{P}^{E'},\v \mathbf{P}^{E'})\\
 &  \leq   \tau(\M_w^{[E' \times E]})\, \dh(\u,\v),
\end{align*}
(the first inequality is Birkhoff's theorem applied to the matrix $\M_w^{[E' \times E]}$ and the second one is the simple observation that projections do not increase $\dh$ distance: this follows directly from the definition of $\dh$ and the fact that $\max_{i \in E} (\u(i)/\v(i)) \geq \max_{i \in E'} (\u(i)/\v(i))$ and $\max_{i \in E} (\v(i)/\u(i)) \geq \max_{i \in E'} (\v(i)/\u(i))$ when $E' \subseteq E$).
We have thus proven our proposition by taking $\tau(E,w)=\tau(\M_w^{[E' \times E]})$.
\end{proof}

Finally, we show that pseudo-mixing words contract the set $\U^+_m(E)$.

\begin{proposition} \label{prop:our-perron-frobenius}
Let $E \subseteq \interval$ and let~$w \in \alph^*$ be a word which pseudo-mixes~$E$. Then, there exists a vector $\x \in \U_m^+(E)$ and a constant $c \geq 0$ such that for every~$k$ and for every $\v \in \U_m^+(E)$,
\[
\dh(\v \cdot w^k, \x) \leq c\, \tau(E,w)^k. 
\]
\end{proposition}

\begin{proof}
Let us define $\mathbf{I}_k$ to be the image of $\U_m^+(E)$ under $w^k$:
\[
\mathbf{I}_k = \{\v \cdot w^k \mid \v \in \U_m^+(E)\}. 
\]
Since for every~$\v$, we have $\v \cdot w^{k+1} = (\v \cdot w) \cdot w^k$, we see that the sequence $(\mathbf{I}_k)$ satisfies $\mathbf{I}_0 \supseteq \mathbf{I}_1 \supseteq\mathbf{I}_2 \ldots$. 

We show that $\mathbf{I}_1$ has finite diameter w.r.t. the distance $\dh$. Let $E'$ be as in the previous proof, $\v$ be any vector in $\U_m^+(E)$, and $j \in E'$ be such that $\v(j) = \max_{i \in E'} \v(i)$. The matrix $\M_w^{[E' \times E]}$ only has positive entries, let us say they all belong to the interval $[\alpha,\beta]$ where $0<\alpha \leq \beta$. Then, every coordinate of $\v \M_w$ indexed by~$E$ is at least $\alpha \v(j)$ and at most $m \beta \v(j)$, hence the ratio between the max and min value is at most $\frac{m \beta}{\alpha}$. Another way to say this is that 
\[
\dh(\v \M_w,\one_E) \leq \ln \left(\frac{m \beta}{\alpha} \right),   
\]
where $\one_E$ is the vector whose coordinates are $1$ on $E$ and $0$ outside~$E$. Since $\v$ is arbitrary, by the triangular inequality this shows that $\mathbf{I}_1$ has $\dh$-diameter at most $2\ln \left(\frac{m \beta}{\alpha} \right) $. Now, applying Proposition~\ref{prop:birkhoff-pseudo-mixing} inductively, we get that for every~$k \geq 1$, $\mathbf{I}_k$ has diameter at most $2\ln \left(\frac{m \beta}{\alpha} \right) \tau(w,E)^{k-1}$, which tends to~$0$ as $k$ tends to $+\infty$. Recall that $\dh$ dominates the $\norm{}$-distance on $\U_m^+(E)$, thus the $\norm{}$-diameter of $\mathbf{I}_k$ must also be bounded by $2\ln \left(\frac{m \beta}{\alpha} \right) \tau(w,E)^{k-1}$. 

The coordinate bounds above imply that $\overline{\mathbf I_1}$ is a compact subset of $\U_m^+(E)$. Thus, the nested compact sets $\overline{\mathbf I_k}$, for $k\geq1$, have diameters tending to zero, so their intersection is a singleton $\{\x\}$. By continuity,
for every $k\geq1$,
\[
\overline{\mathbf I_{k+1}}
= \{\v\cdot w^k \mid \v\in\overline{\mathbf I_1}\}
\subseteq \mathbf I_k.
\]
Hence $\bigcap_{k\geq1}\mathbf I_k=\{\x\}$, and for every $k\geq1$,
we have
\[
\dh(\v \cdot w^k, \x) \leq 2\ln \left(\frac{m \beta}{\alpha} \right) \tau(w,E)^{k-1}, 
\]
and the result follows by taking $c= 2\ln \left(\frac{m \beta}{\alpha} \right) \tau(w,E)^{-1}$. 
\end{proof}

\subsection{A finer statement under the $(\bigstar)$ hypothesis}

We will now state a finer version of Theorem~\ref{thm:main} (under assumption~$(\bigstar)$), by distinguishing three possible scenarios, each of which leads to a different behavior of the sequence $(\v \M_{X[1 \ldots n]})_{n \in \N}$. 

Let us call $\hat{\A}$ the sub-automaton of $\A$ restricted to the states reachable from singleton states. The first possible (and simplest!) scenario is the following. 

\begin{proposition}[Case 0]\label{prop:case_0}
If $\hat{\A}$ only has one BSCC (the null one) then for every normal~$X$ and for every starting vector~$\v$ it holds that $\v \M_{X[1\ldots n]}$ is ultimately~$\zero$.
\end{proposition}

\begin{proof}
This follows immediately from the existence of BSCC-synchronizing words. Namely, by normality the BSCC-synchronizing word appears with positive frequency in $X$, and thus $\v \M_{X[1\ldots n]}=\zero$ for some $n$. 
\end{proof}

If we are not in Case 0, Lemma~\ref{lem:unique-bscc}(b) tells us that $\hat{\A}$ has exactly one non-null BSCC, which we called $\B$ (and $\hat{\A}$ may or may not also have a null BSCC). To present the last two scenarios, we fix some member~$F$ of~$\B$ (for example, the minimal~$F$ we considered in Lemma~\ref{lem:unique-bscc}). We further fix a word~$w$ which pseudo-mixes~$F$ (which exists by Lemma~\ref{lem:unique-bscc}(c)) and call~$\x \in \U_m^+(F)$ the vector that comes out of Proposition~\ref{prop:our-perron-frobenius} when $E=F$. 

\begin{proposition}[Case 1]\label{prop:case_0_statement}
If for some word~$z \in \alph^*$ we have that $\Delta(\x \cdot z) > 0$, then for every normal~$X$ and every starting vector~$\v$, either $\v \M_{X[1 \ldots n]}$ is ultimately~$\zero$, or  $\norm{\v \M_{X[1 \ldots n]}}$ converges exponentially fast to~$0$ while remaining positive (and there are vectors~$\v$ for which the latter happens for every normal~$X$). 
\end{proposition}

In the opposite case, we have a very different behavior. 

\begin{proposition}[Case 2]\label{prop:case_1_statement}
If for every word~$z \in \alph^*$ we have that $\Delta(\x \cdot z) = 0$, then for every normal~$X$ and for every starting vector~$\v$, either $\v \M_{X[1 \ldots n]}$ is ultimately~$\zero$, or  $\norm{\v \M_{X[1 \ldots n]}}$ converges exponentially fast to a positive constant (and there are vectors~$\v$ for which the latter happens for every normal~$X$). 
\end{proposition}

In the following we prove both Propositions~\ref{prop:case_0_statement} and~\ref{prop:case_1_statement}.

\subsection{Case 1 and fast convergence to~$\zero$}

Assume that we are in the above Case 1. From here on, we fix a word~$z$ and $\delta>0$ such that $\Delta(\x \cdot z) > 2\delta$. By Lemma~\ref{lemma:continuity_of_action} we know that $\Delta$ and the dot-product by a fixed word are continuous over $\U_m^+(F)$. Hence, there is some~$\varepsilon>0$ such that $\dh(\x,\x')< \varepsilon$ implies $\Delta(\x' \cdot z)>\delta$. Consequently, by Proposition~\ref{prop:our-perron-frobenius}, there exists a~$k$, which we now fix, such that for all~$\v \in \U_m^+(F)$, 
\[
\Delta(\v \cdot w^k z) > \delta, 
\]
which, by invariance of $\Delta$ under positive scalar multiplication, is equivalent to saying:
\[
\text{for all~$\v$ of support~$F$},\  \Delta(\v \M_{w^k z}) > \delta. 
\]

Equipped with this fact, let us analyze what happens if we choose a starting vector $\e_i$ in the basis and look at the trajectory $(\e_i \M_{Z[1 \ldots n]})_{n \in \N}$ when~$Z \in \alph^\omega$ is an infinite sequence taken uniformly at random. 

Let us first consider the trajectory of the support $\supp(\e_i \M_{Z[1 \ldots n]})$, which is equal to {$\{i\} \cdot Z[1 \ldots n]$}. This amounts to treating $\hat{\A}$ as a Markov chain with starting state $\{i\}$ and where each transition $E \overset{a}{\rightarrow} E'$ has probability $1/|\alph|$. 
Due to the existence of BSCC-synchronizing words we know that with probability 1 a trajectory of a Markov chain eventually reaches a BSCC, and once this happens it holds that, due to Proposition~\ref{prop:automata_behaves_as_markov_chain}, with probability~$1$ each state~$E$ of the BSCC occurs in the trajectory with asymptotic frequency $\pi(E)>0$.
In our case $\hat{\A}$ has at most two BSCC's: $\B$ and possibly the null one. This means that with probability~$1$ over~$Z$, either $\{i\} \cdot Z[1 \ldots n]$ is $\emptyset$ for almost all~$n$, or $\{i\} \cdot Z[1 \ldots n]$ is equal to~$F$ with asymptotic frequency $\pi(F)>0$. 
If the latter happens, by Proposition~\ref{prop:automata_behaves_as_markov_chain}, setting $l=|w^k z|$ we see that, with (conditional) probability 1, with asymptotic frequency $\pi(F)|\alph|^{-l}$ of indices~$n$, we have $\{i\} \cdot Z[1 \ldots n]=F$ \emph{and} $Z[n+1\ldots n+l]=w^k z$. Observe that by definition of $w$ and $k$, for any such~$n$, we must have $\Delta(\e_i \cdot Z[1 \ldots n+l]) > \delta$.

Then, recalling the definition of $\Gamma$ (Definition~\ref{def:Gamma}), we have that with probability~$1$ over~$Z$: 
\[
\text{either}~ \e_i \M_{Z[1 \ldots n]} = \zero~ \text{for almost all~$n$, or}~ \Gamma^{\e_i} (Z[1 \ldots n]) \geq \alpha n~  \text{for almost all~$n$},  
\]
where $\alpha$ is any positive constant smaller than $\pi(F)|\alph|^{-l} \delta$, which we now fix.  From this, the following ``finite version''  immediately follows. 

\begin{claim}\label{claim:gamma_grows_fast}
For any~$\eta>0$, for all~$N$ large enough,
\[
\Pr_{|z|=N} \left[\e_i \M_z = \zero ~\text{or}~ \Gamma^{\e_i} (z) \geq \alpha N \right] \geq 1-\eta. 
\]
\end{claim}

Now, recall that $L^{\e_i} + \Gamma^{\e_i}$ is a supermartingale. Moreover, its increase at each step is bounded since $L^{\e_i}(za) \leq L^{\e_i}(z)+\ln|\alph|$ (because of superfairness) and $\Gamma^{\e_i}(za)-\Gamma^{\e_i}(z)=\Delta(\e_i\M_z)\leq1$\footnote{Here is where the artificial bound on the risk function becomes relevant.}.
Thus, we can apply our one-sided version of Azuma's inequality (Proposition~\ref{prop:subexponential_increase_of_A_martingales}) and obtain the following claim.

\begin{claim}\label{claim:gamma_and_L_grow_subexponential}
For any~$\eta>0$, for all~$N$ large enough,
\[
\Pr_{|z|=N} \left[L^{\e_i}(z) + \Gamma^{\e_i} (z) \leq(\alpha/2) N \right] \geq 1-\eta. 
\]
\end{claim}

\noindent And combining both Claims~\ref{claim:gamma_grows_fast} and~\ref{claim:gamma_and_L_grow_subexponential} we conclude that

\begin{claim}
For any~$\eta>0$, for all~$N$ large enough,
\[
\Pr_{|z|=N} \left[L^{\e_i}(z) \leq -(\alpha/2) N  \right] \geq 1-2\eta. 
\]
In other words, 
\[
\Pr_{|z|=N} \left[\norm{\e_i \M_z} \leq e^{-(\alpha/2)N}  \right] \geq 1-2\eta. 
\]
\end{claim}

\noindent Finally, by applying this last Claim to every canonical vector and using the union bound we obtain

\begin{claim}\label{claim:allmost_all_sequences_lose_exponentially_fast}
For any~$\eta>0$, for all~$N$ large enough,
\begin{align}\label{eq:all_vectors_contract}
	\Pr_{|z|=N} \left[\forall \v\, \norm{\v \M_z} \leq e^{-(\alpha/2)N} \norm{\v}  \right] \geq 1-2m \eta. 
\end{align}
\end{claim}

Using this last result we prove that, if $X$ is a normal sequence, then $\norm{\v \M_{X[1\ldots n]}}$ tends to~$0$ exponentially fast as $n\to \infty$.

Because of normality, if we divide $X$ into blocks of length $N$ as $X[1\ldots N]$, $X[N+1\ldots 2N]$,$\ldots$, $X[N(k-1)+1\ldots Nk]$ it holds that
the fraction of blocks that do not satisfy the constraint in Eq.~\eqref{eq:all_vectors_contract} tends to a value bounded by $2m\eta$ as we increase $k$, and so for almost all values of $k$ this
fraction is bounded by $3m\eta$. For these ``bad'' blocks the increase in norm is bounded by $|\alph|^N$, while for the ``good'' ones we can bound it by $e^{-(\alpha/2)N}$. Thus, for all $k$ big enough, and assuming for simplicity that $\norm{\v}=1$, we conclude that
\begin{align*}
	\norm{\v \M_{X[1\ldots Nk]}} \leq |\alph|^{Nk(3m\eta)} e^{-(\alpha/2)Nk(1-3m\eta)}. 
\end{align*}
Taking the logarithm:
\begin{align*}
	\log \norm{\v \M_{X[1\ldots Nk]}} \leq Nk\left(\log |\alph| (3m\eta) -(\alpha / 2)(1-3m\eta)\right). 
\end{align*}
Thus, for any $\beta < \alpha/2$ we can choose $\eta$ small enough so that for almost all $k$ it holds that
\begin{align*}
	\norm{\v \M_{X[1\ldots Nk]}} \leq e^{-\beta Nk}. 
\end{align*}
Finally, to consider positions in between blocks we recall that $\norm{\v \M_a} \leq |\alph|\,\norm{\v}$, and therefore the increase in norm is at most a constant (because $N$ is already fixed). Then, we conclude that
\begin{align*}
	\norm{\v \M_{X[1\ldots n]}} = O(e^{-\beta n})
\end{align*}
for all $n$. 

\subsection{Case 2 and fast convergence to a constant}

We now address Case 2, that is, when for every word~$z \in \alph^*$ we have that $\Delta(\x \cdot z) = 0$, or equivalently $\norm{\x \M_z}=\norm{\x}$ for all~$z \in \alph^*$.

Let us denote by $\S$ the cone of ``non-betting'' vectors (to which~$\x \cdot z$ belongs for all~$z$):
\[
\S = \left\{\v \in \vect \mid \forall z \in \alph^* \ \norm{\v \M_z}=\norm{\v} \right\}. 
\]

Note that $\S$ is topologically closed (w.r.t. the $\norm{}$-distance), closed under nonnegative scalar multiplication and addition and closed under multiplication by any matrix $\M_z$. 

Suppose a vector $\v$ can be written as a sum $\s + \u$, where $\s \in \S$. Then for any $w$, we have by definition of~$\S$, $\norm{\v \M_w} = \norm{\s} + \norm{\u\M_w}$. If we think of $\norm{\v}$ as capital, then we can view $\s$ as being a part of $\v$ that has been put aside and will never be used in future ``bets''. Only $\u$ can contain ``live'' capital which can be used to place bets. This leads us to the next definition. 

\begin{definition}
Let $\v \in \vect$. We define 
\[
\Live(\v) = \norm{\v} - \max \{\norm{\s} \mid \s \in \S, \, \s \leq \v \}. 
\]
\end{definition}

Note that this definition is sound because the set $\{\s \in \S, \, \s \leq \v \}$ is a compact set for the \norm{}-distance, hence the maximum $\norm{\s}$ is reached. Let us now prove some important properties of the $\Live$ function. 

\begin{proposition} \label{prop:live}
\begin{itemize}
\item[(a)] For any $\v \in \vect$, $\Live(\v) \leq \norm{\v}$. 
\item[(b)] For any $\u \in \vect$ and $\alpha \in \mathbb{R}^{\geq 0}$, $\Live(\alpha \u) = \alpha\, \Live(\u)$
\item[(c)] For any $\u, \v \in \vect$, $\Live(\u + \v) \leq \Live(\u) + \Live(\v)$
\item[(d)] For all $\varepsilon>0$ small enough, if $\norm{\u} = \norm{\v}$ and $\dh(\u,\v) \leq \varepsilon$, then $|\Live(\u)-\Live(\v)| \leq 3\varepsilon \, \norm{\u}$.
\item[(e)] For all~$\v$, for all $a \in \alph$, $\Big| \norm{\v \M_a} - \norm{\v} \Big| \leq |\alph| \, \Live(\v)$.
\end{itemize}
\end{proposition}

\begin{proof}
Item $(a)$ is immediate from the definition. Item $(b)$ follows, for $\alpha>0$, from the invariance of~$\S$ under non-negative scalar multiplication and the fact that $\alpha \s \leq \alpha \v$ if and only if $\s \leq \v$ (the case $\alpha=0$ is trivial). Likewise, $(c)$ follows from the invariance of $\S$ under addition and the fact that $\s \leq \u$ and $\s' \leq \v$ implies $\s + \s' \leq \u + \v$. For $(d)$, because of $(b)$, it suffices to treat the case where $\norm{\u}=\norm{\v}=1$. Let $E$ be the support of $\u$ and $\v$. By definition of $\dh$, having $\dh(\u,\v) \leq \varepsilon$ means in particular that $\max_{i \in E} \u(i)/\v(i) \leq e^\varepsilon$, hence, $e^{-\varepsilon} \u \leq  \v$. Now, let $\s \in \S$ be $\leq \u$ and of maximal norm, so that $\Live(\u) = 1-\norm{\s}$. Since we have $e^{-\varepsilon} \s \leq e^{-\varepsilon} \u \leq \v$, it follows that
\begin{align*}
\Live(\v) &  \leq 1 - e^{-\varepsilon} \norm{\s}\\
 & = 1 - e^{-\varepsilon} (1-\Live(\u)) \\
 & = 1-e^{-\varepsilon} + \Live(\u) - (1-e^{-\varepsilon}) \Live(\u)\\
 & \leq 2(1-e^{-\varepsilon}) + \Live(\u) 
\end{align*}
(the last inequality is due to $\Live(\u) \leq \norm{\u} = 1$). For $\varepsilon$ small enough, we have $2(1-e^{-\varepsilon}) \leq 3\varepsilon$, and the result follows by symmetry. 
For item $(e)$, recall that by superfairness, $|\alph|$ dominates the matrix norms of all~$\M_a$, i.e., $\norm{\u \M_a} \leq |\alph| \norm{\u}$ for all vectors~$\u$. Let $\v$ be a vector and $\s \in \S$ be $\leq \v$ and of maximal norm, so that $\Live(\v) = \norm{\v} -\norm{\s}$. Write $\v = \v' + \s$. For all~$a \in \alph$, we have
\[
\norm{\v \M_a} = \norm{\v' \M_a} + \norm{\s \M_a } =  \norm{\v' \M_a} + \norm{\s} = \norm{\v' \M_a} + \norm{\v} - \norm{\v'}. 
\]
This gives 
$\Big| \norm{\v \M_a} - \norm{\v} \Big| \leq |\alph| \, \norm{\v'} = |\alph| \, \Live(\v)$, as wanted.
\end{proof}

Recall that~$\x \in \U_m^+(F)$ is the vector that comes out of Proposition~\ref{prop:our-perron-frobenius} when $E=F$, and~$w$ is a fixed word which pseudo-mixes~$F$. Our assumption (Case 2) is that $\Delta(\x \cdot z) = 0$ for all~$z$. 

Like in Case 1, we use the fact that if~$Z$ is chosen at random, for any given~$i$, with probability~$1$ over~$Z$, either $\e_i  \cdot Z[1 \ldots n]=\zero$ for almost all~$n$, or $\{i\} \cdot Z[1 \ldots n]$ is equal to~$F$ with asymptotic frequency $\pi(F)>0$. If the latter happens, then with probability~1 (conditionally to being in this second case), with asymptotic frequency $\pi(F)|\alph|^{-|w|}$ of indices~$n$, we have $\{i\} \cdot Z[1 \ldots n]=F$ \emph{and} $Z[n+1\ldots n+|w|]=w$. The first time $F$ is hit in the trajectory of $\{i\}\cdot Z[1 \ldots n]$, at index $n_0$, $\e_i  \cdot Z[1 \ldots n_0]$ is at a certain unknown $\dh$-distance from~$\x$. But if from this point on we compare the trajectories of $\e_i  \cdot Z[1 \ldots n]$ and $\x \cdot Z[n_0+1 \ldots n]$, we know from Lemma~\ref{lem:dh-perservation} that the distance $\dh(\e_i  \cdot Z[1 \ldots n],\x \cdot Z[n_0+1 \ldots n])$ is non-increasing in~$n$ and, moreover, every time the index~$n$ is such that $\{i\}\cdot Z[1 \ldots n]=F$ and $Z[n+1\ldots n+|w|]=w$, by Proposition~\ref{prop:birkhoff-pseudo-mixing}, the distance $\dh(\e_i  \cdot Z[1 \ldots n],\x \cdot Z[n_0+1 \ldots n])$ decreases by a factor $\tau(F,w) \in (0,1)$ (more specifically: $\dh(\e_i  \cdot Z[1 \ldots n+|w|],\x \cdot Z[n_0+1 \ldots n+|w|]) \leq \tau(F,w) \dh(\e_i  \cdot Z[1 \ldots n],\x \cdot Z[n_0+1 \ldots n])$). 

To sum up, with (conditional) probability~$1$ with positive frequency over~$n$, the distance $\dh(\e_i  \cdot Z[1 \ldots n],\x \cdot Z[n_0+1 \ldots n])$ sees a decrease of a factor $\tau(F,w) \in (0,1)$ between index~$n$ and $n+|w|$ and, again, since $\dh(\e_i  \cdot Z[1 \ldots n],\x \cdot Z[n_0+1 \ldots n])$ is non-increasing over~$n$, we have that $\dh(\e_i  \cdot Z[1 \ldots n],\x \cdot Z[n_0+1 \ldots n])$ tends to~$0$ exponentially fast (again, with conditional probability 1). Summing it all up, with probability~$1$ over~$Z$:
\begin{equation} \label{eq:shrinking-distance}
\text{for almost all~$n$, either}~ \e_i \M_{Z[1 \ldots n]} = \zero~ \text{or}~ \dh(\e_i  \cdot Z[1 \ldots n],\x \cdot Z[n_0+1 \ldots n]) \leq e^{-\alpha n}, 
\end{equation}
where $\alpha>0$ is some fixed constant, independent of~$Z$. 

By the definition of $\x$ we have that $\x \cdot Z[n_0+1 \ldots n]$ belongs to $\S$ for all~$n$. Or equivalently, $\Live(\x \cdot Z[n_0+1 \ldots n])=0$ for all~$n$. Using Proposition~\ref{prop:live}(d), we can rewrite~(\ref{eq:shrinking-distance}) and thus get that with probability~$1$ over~$Z$:
 \begin{equation} \label{eq:shrinking-live}
\text{for almost all~$n$, either}~ \e_i \M_{Z[1 \ldots n]} = \zero~ \text{or}~ \Live(\e_i  \cdot Z[1 \ldots n]) \leq 3e^{-\alpha n}. 
\end{equation}
Finally, observe that $\e_i \M_{Z[1 \ldots n]} = \zero$ implies that $\Live(\e_i \M_{Z[1 \ldots n]}) = 0$ (Proposition~\ref{prop:live}(a)) and $\Live(\e_i  \cdot Z[1 \ldots n]) \leq 3e^{-\alpha n}$ is equivalent to $\Live(\e_i  \M_{Z[1 \ldots n]}) \leq 3e^{-\alpha n} \norm{\e_i  \M_{Z[1 \ldots n]}}$ (Proposition~\ref{prop:live}(b)).

Like for Case 1, we use this long-run analysis to immediately derive the following finite version: 

\begin{claim}\label{claim:live1}
For any~$\eta>0$, for all~$N$ large enough,
\[
\Pr_{|z|=N} \left[\Live(\e_i \M_z)\leq 4e^{-\alpha N} \norm{\e_i  \M_{z}}\right] \geq 1-\eta. 
\]
\end{claim}

But we also know by Kolmogorov's inequality (Theorem~\ref{thm:kolmogorov-inequality}) that $\Pr_{|z|=N} [\norm{\e_i  \M_{z}} \geq 1/\eta] \leq \eta$. The above claim then yields the following.

\begin{claim}\label{claim:live2}
For any~$\eta>0$, for all~$N$ large enough,
\[
\Pr_{|z|=N} \left[\Live(\e_i \M_z)\leq (4/\eta) e^{-\alpha N}\right] \geq 1-2\eta. 
\]
\end{claim}

This being true for any single~$i$, we can apply the union bound and get 

\begin{claim}\label{claim:live3}
For any~$\eta>0$, for all~$N$ large enough,
\[
\Pr_{|z|=N} \left[\forall 1 \leq i \leq m,~ \Live(\e_i \M_z)\leq (4/\eta) e^{-\alpha N}\right] \geq 1-2\eta m. 
\]
\end{claim}

By sub-linearity of the $\Live$ function (Proposition~\ref{prop:live}(b, c)), this can further be rewritten as 
\begin{claim}\label{claim:live4}
For any~$\eta>0$, for all~$N$ large enough,
\[
\Pr_{|z|=N} \left[\forall \v~ \Live(\v \M_z)\leq (4/\eta) e^{-\alpha N} \norm{\v} \right] \geq 1-2\eta m. 
\]
\end{claim}

Let $B(N,\eta)$ be the set
\[
B(N,\eta) = \{z \in \alph^N  \mid \forall \v~ \Live(\v \M_z)\leq (4/\eta) e^{-\alpha N} \norm{\v} \}
\]
(which by the above claim has cardinality $\geq |\alph|^N(1-2\eta m))$. We claim that another expression for~$B(N,\eta)$ is 
\[
B(N,\eta) = \{z \in \alph^N \mid \forall \v~ \Live(\v \M_z)\leq (4/\eta) e^{-\alpha N} \Live(\v) \}. 
\]

Indeed, let $z \in B(N,\eta)$ and $\v$ a vector. By definition of $\S$ and $\Live$, we can write $\v$ as $\v=\s+\v'$ with $\s\leq \v$,  $\v' \leq \v$, $\s \in \S$ (equivalently, $\Live(\s)=0$) and $\Live(\v)=\Live(\v')=\norm{\v'}$. We then have:
\begin{itemize}
\item $\Live(\s \M_z)=0$ (by closure of $\S$ under multiplication by matrices of type $\M_{y}$, $y \in \alph^*$)
\item $\Live(\v' \M_z) \leq (4/\eta) e^{-\alpha N} \norm{\v'} = (4/\eta) e^{-\alpha N} \Live(\v)$ (the inequality by definition of $B(N,\eta)$, the equality by the choice of $\v'$)
\end{itemize}
Therefore:
\begin{align}\label{eq:from_norm_to_live}
	\Live(\v \M_z) = \Live(\s\M_z+\v'\M_z) \leq \Live(\s\M_z) + \Live(\v'\M_z) \leq (4/\eta) e^{-\alpha N} \Live(\v). 
\end{align}

This proves our claim about the alternative expression of $B(N,\eta)$ and thus Claim~\ref{claim:live4} can be reformulated one last time. 

\begin{claim}\label{claim:live5}
For any~$\eta>0$, for all~$N$ large enough,
\[
\Pr_{|z|=N} \left[\forall \v~ \Live(\v \M_z)\leq (4/\eta) e^{-\alpha N} \Live(\v) \right] \geq 1-2\eta m. 
\]
\end{claim}

This last claim is the analogue of Claim~\ref{claim:allmost_all_sequences_lose_exponentially_fast} in Case 1, where we argued that most blocks of length~$N$ caused a linear decrease of the norm, leading to the norm of $\v \M_{X[1 \ldots n]}$ to be exponentially decreasing in~$n$ for any normal~$X$. Here we have the same situation, but with the $\Live$ function instead of the norm, so we'll argue instead that $\Live(\v \M_{X[1 \ldots n]})$ decreases exponentially in~$n$, in a similar way as before. \\

Let~$X$ be a normal sequence and $\v$ a starting vector. If we divide~$X$ into blocks of length~$N$, asymptotically, the fraction of blocks among $X[1\ldots N]$, $X[N+1\ldots 2N]$, $\ldots$, $X[N(k-1)+1\ldots Nk]$ that satisfy the constraint in Claim~\ref{claim:live5} will tend to $|B(N,\eta)|/|\alph|^N$, which we know is at least  $1-2\eta m$. So for almost all~$k$, this fraction will be at least $1-3\eta m$. For the blocks~$x$ that are not in $B(N,\eta)$, the corresponding matrix $\M_x$ can increase the norm by a factor at most $|\alph|^N$ (as in Case~$1$), hence it can only increase the ``live'' part of a vector by that same factor. Putting all this together, we must have, for any vector $\v$, which we assume to be unitary for simplicity, and almost all~$k$,
\begin{align*}
	\Live(\v \M_{X[1\ldots Nk]}) & \leq |\alph|^{Nk(3m\eta)} ((4/\eta) e^{-\alpha N})^{k(1-3m\eta)}.  
\end{align*}

Taking the logarithm:
\begin{align*}
	\ln \Live(\v \M_{X[1\ldots Nk]}) & \leq Nk \left(\ln |\alph| (3m \eta) + \frac{\ln(4/\eta)}{N} - \alpha(1-3m\eta) \right). 
\end{align*}
Thus, for any $\beta < \alpha$, we can choose $\eta$ small enough so that, for almost all~$N$, and almost all~$k$
\begin{align*}
	\ln \Live(\v \M_{X[1\ldots Nk]}) & \leq -\beta Nk,  
\end{align*}
i.e., 
\begin{align*}
	\Live(\v \M_{X[1\ldots Nk]}) & \leq e^{-\beta Nk}. 
\end{align*}

And since in between positions that are multiples of~$N$, $\Live(\v \M_{X[1\ldots n+1]})$ can only be $|\alph|$ times greater than $\Live(\v \M_{X[1\ldots n]})$, this shows that 
\begin{align*}
	\Live(\v \M_{X[1\ldots n]}) &= O(e^{-\beta n}). 
\end{align*}

Then, by Proposition~\ref{prop:live}(e), we must have 
\begin{align*}
	\left| \norm{\v \M_{X[1\ldots n+1]}} - \norm{\v \M_{X[1\ldots n]}} \right| \leq |\alph| \Live(\v \M_{X[1\ldots n]}) \leq O(e^{-\beta n}). 
\end{align*}

A sequence of reals $(x_n)$ satisfying $|x_{n+1}-x_n| = O(e^{-\beta n})$ converges exponentially fast to a limit. Indeed, we can write $x_n = x_0 + \sum_{i=0}^{n-1} (x_{i+1} - x_i)$ and since the series $\sum (x_{i+1} - x_i)$ is summable by our hypothesis, it must converge to a limit, and the distance between $x_n$ and this limit is bounded by $\sum_{i=n}^{\infty} |x_{i+1} - x_i| = O(e^{-\beta n})$. So we have finally proven that $\norm{\v \M_{X[1\ldots n]}}$ converges to a limit exponentially fast.

It remains to show that the limit is positive unless the trajectory
becomes $\zero$. Assume that it does not. Choose an occurrence in $X$ of a BSCC-synchronizing
word for $\hat{\A}$, and let $n_0$ be its ending position.
Put $z=X[1\ldots n_0]$. Since
\[
\v\M_z=\sum_{i=1}^m \v(i)\e_i\M_z\neq\zero,
\]
there is an $i\in\interval$ such that $\v(i)>0$ and
$\e_i\M_z\neq\zero$. Let $E=\supp(\e_i\M_z)$.
Since $z$ ends with a BSCC-synchronizing word, $E$ belongs
to the unique non-null BSCC $\B$ of $\hat{\A}$.
Choose a word $t$ with $F\cdot t=E$.
The vector $\s=\x\M_t$ belongs to $\S$ and has support $E$.
Since $\v\M_z$ is positive on $E$, there is a constant $c>0$
such that $c\s\leq\v\M_z$.
Consequently, for every $n\geq n_0$,
\[
\norm{\v\M_{X[1\ldots n]}}
\geq c\norm{\s\M_{X[n_0+1\ldots n]}}
=c\norm{\s}>0.
\]
The limit is therefore positive. This concludes the proof of Case~2.

\subsection{Examples of each Case}

We now show examples of probabilistic betting automata over alphabet $\alph = \{0, 1\}$ whose corresponding matrices belong to Case 0, 1 or 2.

For Case 0 we can consider the (trivial and deterministic) automaton on the left of Figure~\ref{fig:automata_case_0_and_1} defined over the alphabet $\alph=\{0,1\}$. 

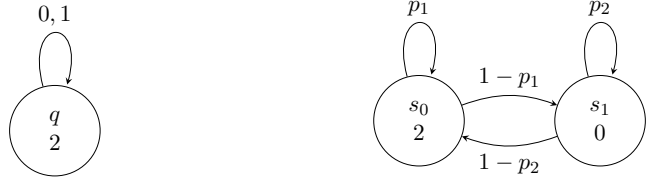
\begin{figure}[ht]
\centering
\resizebox{0.8\textwidth}{!}{%
\begin{minipage}{0.49\textwidth}
\centering
	\begin{tikzpicture}[>=stealth]
		\node[circle, draw, minimum size=1.5cm,align=center] (q) {$q$\\ $2$};
		\draw[->] (q) edge[loop above] node[align=center] {$0, 1$} ();
	\end{tikzpicture}
\end{minipage}
\hfill
\begin{minipage}{0.49\textwidth}
\centering
	\begin{tikzpicture}[>=stealth, node distance=3cm]
		\node[circle, draw, minimum size=1.5cm, align=center] (q0) {$s_0$\\2};
		\node[circle, draw, right of=q0, minimum size=1.5cm, align=center] (q1) {$s_1$\\0};

		\draw[->] (q0) edge[loop above] node {$p_1$} ();
		\draw[->] (q1) edge[loop above] node {$p_2$} ();
		\draw[->] (q0) edge[bend left=20] node[above] {$1-p_1$} (q1);
		\draw[->] (q1) edge[bend left=20] node[below] {$1-p_2$} (q0);
	\end{tikzpicture}
\end{minipage}
}
\caption{On the left, a (deterministic) automaton that belongs to Case 0. Since there is only one state there is only one possible transition. Intuitively, the automaton always bets all of its capital on the next bit being 0.\\
On the right, an automaton that belongs to Case 1 whenever $p_1\neq 1/2$ or $p_2 \neq 1/2$. The transitions are independent of the read bit, and thus we only write the probabilities on the edges (and thus $p_1,p_2 \in (0,1)$). The node on the left ($s_0$) bets all its capital on the next bit being 0, while the one on the right ($s_1$) bets on the next bit being 1.}
\label{fig:automata_case_0_and_1}
\end{figure}

The matrices corresponding to this automaton can be computed following the strategy from the proof of Theorem~\ref{teo:prob_automata_fail_against_normal}, and they are
\begin{align}
	&\M_0 = \begin{pmatrix}
		2
	\end{pmatrix} &\M_1 = \begin{pmatrix}
		0
	\end{pmatrix}.  
\end{align}
The corresponding support graph has two nodes $\{1\}$ and $\emptyset$, with an edge from the former one to the latter one (and self-loops in each node). It is clear that for any normal sequence~$X$ the value $\norm{\v \M_{X[1\ldots n]}}$ will eventually
be 0 for some $n$.

For Case 1, consider the probabilistic automaton on the right of Figure~\ref{fig:automata_case_0_and_1} defined over the same alphabet. The matrices corresponding to this automaton are
\begin{align*}
	&\M_0 = \begin{pmatrix}
		2p_1 & 2(1-p_1)\\
		0 & 0
	\end{pmatrix} &\M_1 = \begin{pmatrix}
		0 & 0\\
		2(1-p_2) & 2p_2
	\end{pmatrix}. 
\end{align*}
The non-null BSCC of the support graph contains uniquely the node $F=\{1,2\}$ with a self-loop. Our hypothesis is that $p_1 \not=1/2$ or $p_2 \not= 1/2$. By symmetry of the automaton, let us assume without loss of generality that the former holds. For any initial vector $\v$ it is the case that $\v \M_a$ is~$\zero$ or has full support, for any $a \in \alph$. A word that pseudo-mixes $F$ is~$0$, since $\M_0$ can already be seen as a projection on the first component followed by a linear operator of dimension $1 \times 2$. Moreover, the vector from Proposition~\ref{prop:our-perron-frobenius}
is $\x = (p_1, 1-p_1)$ and one can readily check that $\norm{\x \M_0}=2p_1$ and $\norm{\x \M_1}=2(1-p_1)$. Therefore, when $p_1 \not= 1/2$, we have $\Delta(\x) > 0$. Thus, we know from our analysis of Case 1 that the norm $\norm{\v \M_{X[1\ldots n]}}$ tends to~$0$ but if the starting vector $\v$ has full support, so does every vector $\v \M_{X[1\ldots n]}$ and thus the limit~$0$ of the norm is never reached.

For Case 2, consider the probabilistic automaton from Figure~\ref{fig:case_2_automata}. The corresponding matrices can be computed as
\begin{align*}
	&\M_0 = \begin{pmatrix}
		1/4 & 5/4 & 1/2 & 0\\
		0 & 0 & 0 & 0\\
		1/2 & 0 & 2/3 & 5/6\\
		0 & 0 & 0 & 0
	\end{pmatrix} &\M_1 = \begin{pmatrix}
		0 & 0 & 0 & 0\\
		1/5 & 1 & 4/5 & 0\\
		0 & 0 & 0 & 0\\
		3/5 & 0 & 2/5 & 1
	\end{pmatrix}. 
\end{align*}
The non-null BSCC of the support graph contains uniquely the node $F =\{1,2,3,4\}$ with a self-loop. Moreover, note that from any initial vector $\v$ the trajectory of the support arrives at $F$ after reading
the word $10$ (unless the vector becomes $\zero$, in which case it arrives to the null BSCC). Thus, for any normal sequence it is clear that $\v\M_{X[1\ldots n]}$ will eventually either have full support or become $\zero$.

The word $10$ pseudo-mixes $F$, since
\begin{align*}
	\M_{10} = \begin{pmatrix}
		0 & 0 & 0 & 0\\
		9/20 & 1/4 & 19/30 & 2/3\\
		0 & 0 & 0 & 0\\
		7/20 & 3/4 & 17/30 & 1/3
	\end{pmatrix}
\end{align*}
is equivalent to a projection over the second and fourth component followed by a positive linear operation of dimension $2 \times 4$. It can be seen that the vector from Proposition~\ref{prop:our-perron-frobenius} is
\begin{align*}
	\x = (1/5, 1/4, 3/10, 1/4),	 
\end{align*}
which is non-betting (i.e. $\Delta(\x\cdot w) = 0$ for every word $w$). Moreover, it can be shown that this is the only unitary stable
configuration (i.e. the only unitary vector $\y$ for which~$\Delta(\y\cdot w)=0$ for all words $w$).
Finally, it can be seen that for almost all initial vectors $\v$ with full support it is the case that $\v \M_w \notin \langle \x \rangle$ for all $w \in \alph^*$. Therefore, there are initial vectors for which the capital never stabilizes but rather converges
to a finite value.


\begin{figure}[ht]
	\centering
	\resizebox{0.7\textwidth}{!}{%
	\begin{tikzpicture}[>=stealth,thick, scale=0.8]

		\tikzset{
		state/.style={draw,circle,minimum size=1.9cm,align=center},
		accepting/.style={draw,double,circle,minimum size=1.9cm,align=center}
		}

		\node[state]     (q0) at (-0.5,2.5) {$s_0$\\ $2$};
		\node[state]     (q2) at (6.5,2.5) {$s_1$\\ $0$};
		\node[state]     (q1) at (-0.5,-2.5) {$s_2$\\ $2$};
		\node[state] (q3) at (6.5,-2.5) {$s_3$\\ $0$};

		\draw[->] (q0) edge[loop left] node {$1/8$} ();
		\draw[->] (q1) edge[loop left] node {$1/3$} ();
		\draw[->] (q2) edge[loop right] node {$1/2$} ();
		\draw[->] (q3) edge[loop right] node {$1/2$} ();

		\draw[->] (q0) to[bend left=12] node[above,pos=0.5] {$5/8$} (q2);
		\draw[->] (q2) to[bend left=12] node[below,pos=0.5] {$1/10$} (q0);

		\draw[->] (q1) to[bend left=12] node[above,pos=0.5] {$5/12$} (q3);
		\draw[->] (q3) to[bend left=12] node[below,pos=0.5] {$1/5$} (q1);

		\draw[->] (q0) to[bend left=18] node[right,pos=0.5] {$1/4$} (q1);
		\draw[->] (q1) to[bend left=18] node[left,pos=0.5] {$1/4$} (q0);

		\draw[->] (q3) -- node[pos=0.3,right,xshift=2mm] {$3/10$} (q0);
		\draw[->] (q2) -- node[pos=0.3,right,xshift=2mm] {$2/5$} (q1);

	\end{tikzpicture}
	}
	\caption{Automaton that belongs to Case 2. The transitions are independent of the read bit, and thus we only write the probabilities on the edges. The nodes on the left ($s_0$ and $s_2$) bet all their capital on the next bit being 0, while those on the right ($s_1$ and $s_3$) bet on the next bit being 1.}
	\label{fig:case_2_automata}
\end{figure}
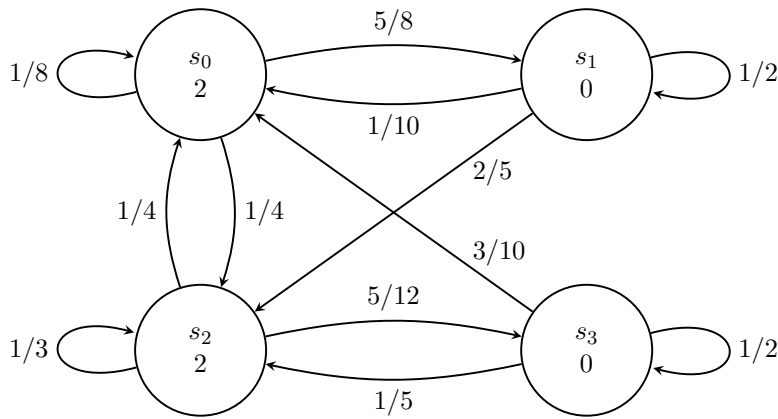

\subsection{Additional remarks (with $(\bigstar)$ hypothesis)}

\subsubsection{Unifying Cases 0, 1 and 2}

We now state a proposition that unifies the different insights developed for each Case.

\begin{proposition}
	For every superfair family of matrices $\{\M_a\}_{a \in \alph}$ satisfying hypothesis $(\bigstar)$ there exists some $\alpha > 0$ such that, for any $\eta > 0$ and $N$ large enough,
	\begin{align*}
		\Pr_{|z| = N}\left[ \forall \v\, \Live(\v \M_z) \leq e^{-\alpha N} \Live(\v) \right] \geq 1 - \eta. 
	\end{align*} 
\end{proposition}

\begin{proof}
	For families of matrices in Case 0 this holds immediately because for $N$ big enough it is the case that $\norm{\v\M_z }=0$ for most sequences $z$ and all vectors $\v$ due to the existence of BSCC synchronizing words. Meanwhile, for those families in Case~2, apply Claim~\ref{claim:live5} with error parameter $\eta/(2m)$. Writing $\alpha_0>0$ for the exponent in that claim, we have $(8m/\eta)e^{-\alpha_0N}\leq e^{-(\alpha_0/2)N}$ for all sufficiently large $N$, so we may take $\alpha=\alpha_0/2$. Lastly, for those in Case 1 this is a consequence of Claim~\ref{claim:allmost_all_sequences_lose_exponentially_fast} (here, too, we apply the claim with error parameter $\eta/(2m)$) and the fact that in this context it is the case that $\Live(\v) = \norm{\v}$ for every $\v$ (otherwise, there would be some initial vectors for which the norm does not converge to zero).
\end{proof}

\subsubsection{Unfair families of matrices and leakage}

Case 2 corresponds to a situation in which the martingale stabilizes because the amount of capital ``in play'' converges exponentially fast to zero. We observe that this cannot happen when the martingale is unfair.

\begin{proposition}\label{prop:unfair_implies_loss}
	Let $\{\M_a\}_{a \in \alph}$ be a family of matrices which satisfies $(\bigstar)$ and is superfair but not fair. Namely, there is some vector $\v$ such that
	\begin{align*}
		\norm{\v} > \frac{\sum_{a \in \alph} \norm{\v \M_a}}{|\alph|}.	 
	\end{align*} 
	Then, Case 2 cannot happen for this family of matrices. 
\end{proposition}

\begin{proof}
	First, we observe that under the hypothesis there exists some canonical vector $\e_i$ that serves as the witness for the unfairness of $\{\M_a\}_{a \in \alph}$. This follows by linearity: if for all $\e_i$ it is the case that $\sum_{a\in \alph} \norm{\e_i \M_a} = |\alph|$ then this equality also holds for any other vector. Thus, we may assume that $\v = \e_i$ for some $i$, which we now fix.

If we are in Case~$0$, the conclusion is immediate so let $F \subseteq [1\ldots m]$ be a set in the non-null BSCC of the support graph, and let $\x$ be the corresponding vector obtainable using Proposition~\ref{prop:our-perron-frobenius}. Because of the assumption $(\bigstar)$ there is some word $z$ such that $i \in F \cdot z$. Letting $\y = \x \M_z$, note that
	\begin{align*}
		\sum_{a \in \alph} \norm{\y \M_{a}} &= \sum_{a \in \alph} \sum_{j=1}^m  \y(j)\norm{\e_j \M_{a}}\\
		&= \sum_{j=1}^m \y(j) \sum_{a \in \alph} \norm{\e_j \M_{a}} < |\alph| \,\,\norm{\y}, 
	\end{align*}
	where the last inequality holds because $\y(i) > 0$. This implies that $\Delta(\y)= \Delta(\x \cdot z) > 0$, which shows that we are not in Case 2.
\end{proof}

Proposition~\ref{prop:unfair_implies_loss} shows that if the game is ``unfair'' to the martingale then its value will always converge exponentially fast to zero (or will even be ultimately~$0$). This fact will be useful when considering the general case
without the $(\bigstar)$ hypothesis. 

\subsubsection{Decidability of the case disjunction}

We show that it is possible to decide which one of the three cases holds if the matrix entries are algebraic using the existential theory of the reals \cite{canny1988some}. The following lemma connects the Perron-Frobenius theorem with the field of algebraic reals. 

\begin{lemma}\label{lemma:computable-perron-frobenius}
Let $\M \in \R^{m \times m}_{> 0}$ be a positive matrix with algebraic entries. Then,
\begin{enumerate}
	\item $\M$ has a unique maximum (in absolute value) eigenvalue $\lambda \in \R_{\geq 0}$ whose related eigenspace $E_{\lambda} = \langle \v \rangle$ is one-dimensional and such that $\v > \zero$ and $\norm{\v}=1$.
	\item For every other eigenvector $\w$ of $\M$ such that $\w \notin E_{\lambda}$ it is the case that there is some $i \in [1\ldots m]$ with $\w_i \notin \R_{> 0}$.
	\item There is a predicate $\Phi_{\v}(\x)$ over the existential theory of the reals that captures $\v$ (i.e. $\Phi_{\v}(\x)$ is true if and only if $\x = \v$)
\end{enumerate}
\end{lemma}

\begin{proof}
The Perron-Frobenius Theorem \cite{horn2012matrix} implies 1) and 2). For the predicate we can take
\begin{align*}
\Phi_{\v}(\x)= \left[\x>\zero\wedge\norm{\x}=1
\wedge\exists\lambda>0\;(\x\M=\lambda\x)\right],
\end{align*}
which is correct due to 1) and 2).
\end{proof}

Lemma~\ref{lemma:computable-perron-frobenius} provides an algorithm to compute the vector $\x$ from Proposition~\ref{prop:our-perron-frobenius} in \PSPACE{}\footnote{Because the existential theory of the reals can be decided in \PSPACE{}\cite{canny1988some}.}. Thus, the only remaining tasks
are to find a node $E\in\B$, a pseudo-mixing word $w$ for $E$ and to decide whether $\Delta(\x \cdot z) > 0$ for some word $z$.

\begin{theorem}\label{teo:decidability}
One can decide, given a superfair family of matrices $(\M_a)_{a\in\alph}$ with algebraic entries, which of the three cases (Case 0, Case 1 or Case 2) holds. Under the standing assumption $(\bigstar)$, these cases are mutually exclusive.
\end{theorem}

\begin{proof}
We first recall that we can decide whether two nodes of the support graph are connected in $\PSPACE{}$ by using the reachability algorithm from Savitch's Theorem~\cite{arora2009computational}. Thus, Case 0 can be detected trivially in $\PSPACE{}$. From now on we focus on the other two cases.

Let $E$ be a node in the unique non-null BSCC $\B$ of the support graph. We can find such an $E$ in $\PSPACE$ by exploiting again the algorithm for reachability that acts implicitly on the support graph.
Moreover, observe that the word $w$ that pseudo-mixes $E$ was constructed explicitly in Lemma~\ref{lem:unique-bscc}, and it has exponential size on the dimension of the matrices\footnote{This is the bottleneck of our algorithm. If it were the case that there always are pseudo-mixing words of length $O(m^k)$ for some $k \in \mathbb{N}$ we would be able to distinguish the cases in $\PSPACE{}$.}. 

Let $E'$ be the indices of the nonzero rows of $\M_w^{[E\times E]}$, and set 
\[
\mathbf B=\M_w^{[E'\times E]},
\qquad
\M_w'=\M_w^{[E'\times E']}.
\]
The matrix $\M_w'$ is strictly positive. By Lemma~\ref{lemma:computable-perron-frobenius}, its normalized positive eigenvector $\v$ is characterized by an existential formula. The vector~$\x$ from Proposition~\ref{prop:our-perron-frobenius} is the zero extension outside $E$ of
\[
\frac{\v\mathbf B}{\norm{\v\mathbf B}}.
\]

Now, to detect whether there is some word $w$ such that $\Delta(\x \cdot w) > 0$, recall that $\Delta(\x\cdot w)=0$ if and only if $\norm{\x\M_w}=\norm{\x\M_{wa}}$ for every $a\in\alph$. Thus, $\Delta(\x \cdot w) = 0$ for all $w$ if and only if, for all $w$ and $a\in \alph$, it holds that
\newcommand{\aOne}{\mathbf{1}}
\begin{align*}
\norm{\x\M_w}=\norm{\x\M_{wa}}
&\iff \x\M_w\aOne=\x\M_w\M_a\aOne\\
&\iff \x \M_w (\mathbf{I} - \M_a) \aOne = 0. 
\end{align*}
Let $\aOne_a = (\mathbf{I} - \M_a) \aOne$. To check this condition we compute a basis for the subspace $S_a = \mathrm{span}(\{\M_w \aOne_a : w \in \alph^*\})$. We define
\begin{align*}
S_a^k = \mathrm{span}(\{\M_w \aOne_a : w \in \alph^{\leq k}\})
\end{align*}
and observe that $S_a = S_a^{m-1}$. To prove this, we show by induction on $i$ that $S_a^{m-1+i} \subseteq S_a^{m-1}$. For $i=1$ take $w \in \alph^{m}$ and note that, because $dim(S_a) \leq m$, there is some $j \in \interval$ such that
$\M_{w[j \ldots m]} \aOne_a \in S_a^{m-j}$: note that $dim(S_a^0) = 1$, and thus it cannot be the case that each of the $m$ vectors $\M_{w[j \ldots m]} \aOne_a$ increases the dimension of the set $S_a^{m-j}$ when added to it. Then, $\M_{w[j \ldots m]} \aOne_a  = \sum_{w' \in \alph^{\leq m-j}} \alpha_{w'} \M_{w'} \aOne_a$, and therefore
\begin{align*}
\M_{w} \aOne_a &= \M_{w[1\ldots j-1]} \M_{w[j \ldots m]} \aOne_a\\
&= \M_{w[1\ldots j-1]} \sum_{w' \in \alph^{\leq m-j}} \alpha_{w'} \M_{w'} \aOne_a\\
&= \sum_{w' \in \alph^{\leq m-j}} \alpha_{w'} \M_{w[1\ldots j-1]} \M_{w'} \aOne_a \in S_a^{(j-1) + (m-j)} = S_a^{m-1}. 
\end{align*}
For the general case, we use a similar strategy: take $bw$ with $b \in \alph$ and $w \in \alph^{m-1+i}$ and note that
\begin{align*}
\M_b \M_{w} \aOne_a &= \M_b \left(\sum_{w' \in \alph^{\leq m-1}} \alpha_{w'} \M_{w'} \aOne_a\right)\\
&= \sum_{w' \in \alph^{\leq m-1}} \alpha_{w'} \M_b \M_{w'} \aOne_a \in S_a^m = S_a^{m-1}. 
\end{align*}

Then, to check if $\Delta(\x \cdot w) = 0$ for all $w$, we check that $\x$ is annihilated by all the subspaces $S_a^{m-1}$ for $a\in \alph$. This can be done within a formula of the existential theory of the reals in which we can use the predicate from Lemma~\ref{lemma:computable-perron-frobenius} to capture $\x$.
\end{proof}

\section{The general case: removing hypothesis~$(\bigstar)$}

Now we prove Theorem~\ref{thm:main} without assuming the $(\bigstar)$ hypothesis. To do this, we need to introduce one last object: the reachability graph.

\newcommand{\aFamily}{\mathcal{M}}
\newcommand{\numberSCCs}{s}

Given a superfair family of non-negative matrices $\aFamily = \{\M_a\}_{a \in \alph}$ where each matrix has dimension $m \times m$ we consider the reachability graph $G_{\aFamily}$ whose nodes are the indices $[1\ldots m]$ and there is an edge $i \to j$ if and 
only if there is some matrix $\M_a$ such that $\M_a(i,j) > 0$. To prove the general theorem we will use an inductive argument on the number $\numberSCCs$ of strongly connected components of $G_{\mathcal{M}}$.
More precisely, we will show inductively on $s$ that 
\begin{claim}\label{claim:convergence_for_s_SCCs}
	Let $\{\M_a\}_{a \in \alph}$ be a superfair family of matrices whose reachability graph has $s$ strongly connected components. Then, there exists some $\alpha > 0$ such that, for every $\eta > 0$ and $N$ big enough,
	\begin{align}\label{eq:claim_convergence_s_SCCs}
		\Pr_{|z| = N}\left[ \forall \v\, \Live(\v \M_z) \leq e^{-\alpha N} \Live(\v) \right] \geq 1 - \eta.  
	\end{align}  
\end{claim}
From this claim the fact that $\norm{\v \M_{X[1\ldots n]}}$ converges exponentially fast to some finite value for every initial vector $\v$ and normal sequence $X$ follows using the same techniques as before. Note that under the $(\bigstar)$ hypothesis this graph is strongly connected, and thus we have already shown the base case $\numberSCCs = 1$.
We prove the general case below.

\newcommand{\firstComponent}{C}

\begin{proof}
	Fix some family $\aFamily = \{\M_a\}_{a \in \alph}$ with $\numberSCCs + 1$ strongly connected components, and sort the components topologically as $\firstComponent, D_1, \ldots, D_s$ (i.e. there are no edges going into $\firstComponent$ and if there is some edge from a node of $D_i$ to a node of $D_j$ then $i \leq j$).

	Let $D = \bigcup_{i=1}^s D_i$. Note that if $\firstComponent$ has no outgoing edges to $D$ then the result follows immediately considering the subfamilies $\{\M_a^{[\firstComponent \times \firstComponent]}\}_{a \in \alph}$ and $\{\M_{a}^{[D \times D]}\}_{a \in \alph}$ which each have at most~$s$ strongly connected components.
	More precisely, we apply Claim~\ref{claim:convergence_for_s_SCCs} inductively to $\firstComponent$ and $D$ to obtain $\alpha_\firstComponent$ and $\alpha_D$. Then, we claim that Equation~\eqref{eq:claim_convergence_s_SCCs} holds for $\aFamily$ with $\alpha = \min\{\alpha_\firstComponent, \alpha_D\}$: for any $\eta > 0$ we can take an $N$ big enough so that both
	\begin{align*}
		\Pr_{|z| = N}\left[ \forall \v\, \Live(\v \M_z^{[C\times C]}) \leq e^{-\alpha_\firstComponent N} \Live(\v) \right] \geq 1 - \eta/2
	\end{align*}
	and
	\begin{align*}	
		\Pr_{|z| = N}\left[ \forall \v\, \Live(\v \M_z^{[D\times D]}) \leq e^{-\alpha_D N} \Live(\v) \right] \geq 1 - \eta/2
	\end{align*}
	hold\footnote{Note that the live capital function $\Live$ depends on the stable cone, which may not be the same on both components, but we simply use $\Live$ for both to avoid heavy notation.}. By the union bound and Proposition~\ref{prop:live}(a,c), with probability at least $1-\eta$ we have $\Live(\v\M_z)\leq e^{-\alpha N}\norm{\v}$ for every $\v$.
Choose $\s\in\S$ with $\s\leq\v$ and $\norm{\v-\s}=\Live(\v)$. Since $\s\M_z\in\S$, applying this bound to $\v-\s$ gives 
\[
\Live(\v\M_z)\leq\Live((\v-\s)\M_z)
\leq e^{-\alpha N}\Live(\v).
\]
Thus,
\begin{align*}	
\Pr_{|z| = N}\left[ \forall \v\, \Live(\v \M_z) \leq e^{-\alpha N} \Live(\v) \right] \geq 1 - \eta. 
\end{align*}

	Therefore, from now on we assume that $\firstComponent$ has some outgoing edges to $D$, i.e. that there is some \textit{leakage} from $\firstComponent$ to $D$.

	Let $\v$ be some initial vector. We will denote by $\v_F$ the projection of $\v$ into the indices from $F \subseteq [1\ldots m]$. In particular, note that $\v = \v_\firstComponent + \v_D$. 
	We will consider a two-step process in which we first evolve $\v$ as $\v\M_{z_1}$ according to some random sequence $z_1 \in \alph^N$ and then again using a second
	random sequence $z_2 \in \alph^N$. We will argue that, during the first step, $\norm{(\v \M_{z_1})_\firstComponent}$ converges exponentially fast to zero with high probability, while $\norm{(\v \M_{z_1})_D}$ does not increase too much (i.e. it grows subexponentially). Later, in the second step
	we will control the live capital of $(\v \M_{z_1})_D \M_{z_2}$ using the inductive hypothesis while arguing that $\norm{(\v \M_{z_1})_\firstComponent \M_{z_2}}$ cannot grow fast enough to compensate the capital lost on the first step (again, with high probability). See Figure~\ref{fig:sketch_of_proof} for a sketch of this proof.

	Consider the family $\{\M_{a}^{[\firstComponent \times \firstComponent]}\}_{a \in \alph}$ obtained by restricting $\aFamily$ to the indices from $\firstComponent$. This family is superfair, but not fair: this stems from the fact
	that there is some edge from $\firstComponent$ to $D$ in the reachability graph, and thus there is some matrix $\M_a$ and indices $i\in \firstComponent,j\in D$ such that $\M_a(i,j) > 0$. Hence, by Proposition~\ref{prop:unfair_implies_loss} this family loses exponentially fast.
	More precisely, the family belongs to Case 0 or Case 1 described above, and thus there is some $\alpha_\firstComponent$ such that for any $\eta>0$ and big enough $N$ it holds that

	\begin{align}\label{eq:control_first_component_first_phase}	
		\Pr_{|z_1| = N}\left[ \forall \v\  \norm{\v_\firstComponent \M_{z_1}^{[\firstComponent\times\firstComponent]}} \leq e^{-\alpha_\firstComponent N} \norm{\v_\firstComponent}\right] \geq 1-\eta.
	\end{align}
	Observe that $\v_\firstComponent \M_{z_1}^{[\firstComponent\times\firstComponent]} =(\v\M_{z_1})_\firstComponent$ because there are no edges from $D$ to $\firstComponent$. Thus, Eq.~\eqref{eq:control_first_component_first_phase} lets us control $\norm{(\v\M_{z_1})_\firstComponent}$ for most sequences
	$z_1$ during our first step. To bound $\norm{(\v \M_{z_1})_{D}}$ we use Theorem~\ref{thm:kolmogorov-inequality}, which states that for any $\eta > 0$ and big enough $N$ it holds that
	\begin{align}\label{eq:control_second_component_first_phase}
		\Pr_{|z_1| = N}\Big[ \forall \v\, \norm{\v \M_{z_1}} \leq (1/\eta)\norm{\v} \Big] \geq 1 - m\eta. 
	\end{align}
	Since $\norm{(\v \M_{z_1})_D} \leq \norm{\v \M_{z_1}}$ Eq.~\eqref{eq:control_second_component_first_phase} is enough for our purposes.

	For the second step, we will control $\norm{(\v\M_{z_1})_\firstComponent \M_{z_2}}$ by using again Theorem~\ref{thm:kolmogorov-inequality}, while we control $\Live((\v\M_{z_1})_D\M_{z_2})$ using the inductive hypothesis. Namely, there is some $\alpha_D$ such that for any $\eta > 0$ and $N$ big enough
	\begin{align}\label{eq:control_second_component_second_phase}	
		\Pr_{|z_2| = N}\left[ \forall \w\, \Live(\w \M_{z_2}^{[D\times D]}) \leq e^{-\alpha_D N} \Live(\w) \right] \geq 1 - \eta.	 
	\end{align}
	Note that if $\w$ only has support on $D$ then $\w \M_{z_2}^{[D\times D]} = \w \M_{z_2}$ since there are no edges from $D$ to $C$, and therefore Eq.~\eqref{eq:control_second_component_second_phase} can be used to bound our expression. 

	Now we mix all these observations in a final expression. Looking at a set of sequences $z =z_1z_2 \in \alph^{2N}$ such that all the above bounds hold and
$\norm{\u\M_{z_2}}\leq(1/\eta)\norm{\u}$ for every $\u$ (which happens with probability at least $1-2(m+1)\eta$) it is the case that for any vector~$\w$
	\begin{align*}
		\Live(\w \M_{z_2}) &\leq \Live(\w_\firstComponent \M_{z_2}) + \Live(\w_D \M_{z_2})\\
		&\leq \norm{\w_\firstComponent \M_{z_2}} + \Live(\w_D \M_{z_2})\\
		&\leq \norm{\w_\firstComponent}/\eta + e^{-\alpha_D N} \Live(\w_D), 
	\end{align*}
	where the first and second inequalities follow by properties of the $\Live$ function and the last inequality by Theorem~\ref{thm:kolmogorov-inequality} and Eq.~\eqref{eq:control_second_component_second_phase}.
	
	If we now instantiate, in the above, $\w = \v \M_{z_1}$ for an arbitrary vector $\v$, it holds that
\begin{align*}
\Live(\v \M_{z_1} \M_{z_2})
&\leq (1/\eta)\norm{(\v\M_{z_1})_\firstComponent}
  + e^{-\alpha_D N}\Live((\v\M_{z_1})_D)\\
&\leq (1/\eta)\norm{(\v\M_{z_1})_\firstComponent}
  + e^{-\alpha_D N}\norm{(\v\M_{z_1})_D}\\
&\leq (1/\eta)e^{-\alpha_C N}\norm{\v_C}
  + (1/\eta)e^{-\alpha_D N}\norm{\v}\\
&\leq (2/\eta)e^{-\alpha N}\norm{\v},
\end{align*}

where the second inequality follows again by properties of the $\Live$ function, the third one by Eqs.~\eqref{eq:control_first_component_first_phase} and~\eqref{eq:control_second_component_first_phase} and the last one by considering $\alpha = \min\{\alpha_C,\alpha_D\}$ and using $\norm{\v_\firstComponent}\leq\norm{\v}$.
Finally, using the same ideas from the proof of Case 2 (more specifically, from Eq.~\eqref{eq:from_norm_to_live}) we can turn this expression into an inequality of the form

\begin{align*}
\Live(\v\M_{z_1z_2})
\leq (2/\eta)e^{-\alpha N}\Live(\v).
\end{align*}

	Given $\varepsilon>0$, take $\eta=\varepsilon/(2(m+1))$.
For all sufficiently large $N$, we have
$(2/\eta)e^{-\alpha N}\leq e^{-(\alpha/2)N}$, and hence
\[
\Pr_{|z|=2N}\left[
\forall\v,\ 
\Live(\v\M_z)\leq e^{-(\alpha/2)N}\Live(\v)
\right]\geq1-\varepsilon.
\]
This proves Claim~\ref{claim:convergence_for_s_SCCs}
for even lengths. The estimate extends to odd lengths after
decreasing the exponential rate, since
$\Live(\u\M_a)\leq|\alph|\Live(\u)$
by the same decomposition as in the proof of
Proposition~\ref{prop:live}(e).

\end{proof}

\begin{figure}
	\centering
	\begin{tikzpicture}[>=stealth,thick]

		\node (V) at (0,2) {$\v$};

		\node (W) at (-2,-0.35) {$\underbrace{(\v\M_{z_1})_\firstComponent}_{\v_\firstComponent \M_{z_1}^{[\firstComponent\times\firstComponent]}}$};
		\node (Z) at (2,0) {$(\v\M_{z_1})_D$};

		\node (X) at (-2,-2.5) {$(\v\M_{z_1})_\firstComponent \M_{z_2}$};
		\node (Y) at (2,-2.85) {$\underbrace{(\v\M_{z_1})_D \M_{z_2}}_{(\v\M_{z_1})_D \M_{z_2}^{[D \times D]}}$};

		\node at (0,0) {$+$};
		\node at (0,-2.5) {$+$};

		\draw[->] (V) to[out=-150,in=90] node[midway,left,xshift=-3mm]{Prop.~\ref{prop:unfair_implies_loss}} (W);
		\draw[->] (V) to[out=-30,in=90] node[midway,right,xshift=3mm]{Thm.~\ref{thm:kolmogorov-inequality}} (Z);

		\draw[->] (W) -- node[left]{Thm.~\ref{thm:kolmogorov-inequality}} (X);
		\draw[->] (Z) -- node[right, pos=0.7]{IH} (Y);

		\draw[->, dashed] (4,1.8) -- node[right]{$z_1$} (4,0.3);
		\draw[->, dashed] (4,-0.5) -- node[right]{$z_2$} (4,-2);

	\end{tikzpicture}
	\caption{Sketch of the proof of Claim~\ref{claim:convergence_for_s_SCCs}. The dashed arrows indicate the evolution of the vector $\v$ according to $z_1$ and $z_2$. The solid arrows indicate which previous results we use to bound the norm or live capital. IH denotes the inductive hypothesis.}
	\label{fig:sketch_of_proof}
\end{figure}
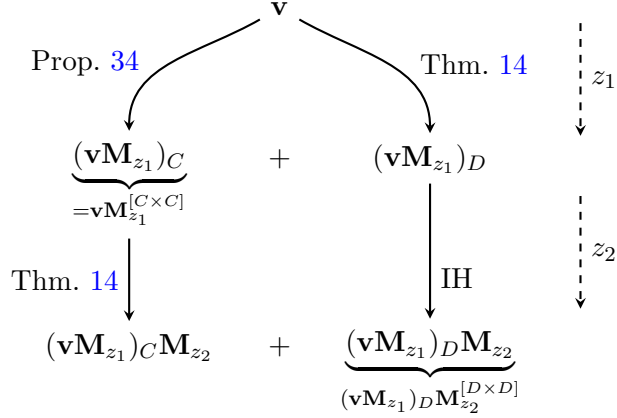

We observe that in the general case the distinction between Case 0, 1 and 2 is not as simple as before. Namely, if there is one BSCC belonging to Case 1 and another BSCC belonging to Case 2 then both behaviours can be attainable depending on the prefix of the normal sequence.

\section{Conclusion and open questions}

In this work we showed that martingales given by the expected value of probabilistic automata do not succeed against normal sequences (Corollary~\ref{coro:expected_success_and_normality}). In particular, we showed that for strongly connected automata, the martingales they represent can behave in three different ways. In Case~0, every martingale reaches zero at some finite step (Proposition~\ref{prop:case_0}). In Case~1, every martingale that remains positive
converges exponentially fast to zero (Proposition~\ref{prop:case_0_statement}). In Case~2, every martingale that remains positive converges exponentially fast to a positive finite limit (Proposition~\ref{prop:case_1_statement}). Moreover, for algebraic transition probabilities and betting factors, the case disjunction is decidable (Theorem~\ref{teo:decidability}).

We leave a concrete open question. Our procedure to decide to which case a given family of matrices/automaton belongs runs in \textsc{EXPSPACE}, even though most of the steps of the algorithm belong to \PSPACE{}. The bottleneck comes from the fact that our construction of pseudo-mixing words in Lemma~\ref{lem:unique-bscc} obtains words whose length
might be exponential in the dimension of the matrices. A proof that there are pseudo-mixing words of polynomial length would immediately imply a $\PSPACE{}$ algorithm to solve this problem.

\bibliographystyle{plain}
\bibliography{biblio}

\end{document}